\documentclass{article}

\usepackage{arxiv}

\usepackage[utf8]{inputenc} 
\usepackage[T1]{fontenc}    
\usepackage[hidelinks]{hyperref}       
\usepackage{url}            
\usepackage{booktabs}       
\usepackage{amsfonts}       
\usepackage{nicefrac}       
\usepackage{microtype}      
\usepackage{graphicx}
\usepackage{natbib}
\usepackage{doi}
\usepackage{authblk}
\usepackage{amsthm} 
\usepackage[mathscr]{euscript}

\newtheorem{theorem}{Theorem}
\newtheorem{proposition}{Proposition}
\newtheorem{lemma}{Lemma}

\newtheorem{assumption}{Assumption}

\usepackage{tikz}
\usetikzlibrary{positioning, shapes.geometric, calc, decorations.markings}

\tikzset{
    flowbox/.style={rectangle, draw, minimum width=2.5cm, minimum height=1cm, text width=2.3cm, align=center, fill=gray!15},
    flowarrow/.style={->, very thick},
    decision/.style={diamond, draw, minimum width=2.5cm, minimum height=1.5cm, text width=2cm, align=center, aspect=2},
    choice/.style={rectangle, draw, rounded corners=0.3cm, minimum width=2.5cm, minimum height=1cm, text width=2.3cm, align=center, fill=yellow!20},
    inputbox/.style={trapezium, draw, trapezium left angle=70, trapezium right angle=110, minimum width=1.5cm, minimum height=1cm, inner xsep=0.3cm, inner ysep=0.2cm, align=center, fill=blue!10, text width=2.3cm},
    outputbox/.style={trapezium, draw, trapezium left angle=70, trapezium right angle=110, minimum width=1.5cm, minimum height=1cm, inner xsep=0.3cm, inner ysep=0.2cm, align=center, fill=green!10, text width=2.3cm},
    input/.style={rectangle, draw, minimum width=1.5cm, minimum height=1cm, inner xsep=0.3cm, inner ysep=0.2cm, align=center, fill=blue!15, thick, rounded corners=2pt},
    output/.style={rectangle, draw, minimum width=1.5cm, minimum height=1cm, inner xsep=0.3cm, inner ysep=0.2cm, align=center, fill=green!15, thick, rounded corners=2pt},
    yes/.style={flowarrow, edge label={node[midway, above, font=\small]{Yes}}},
    no/.style={flowarrow, edge label={node[midway, below, font=\small]{No}}},
    param/.style={flowarrow},
    circlenode/.style={circle, draw, minimum size=0.75cm, inner sep = 1pt},
    recordednode/.style={circlenode, fill=gray!30},
    latentnode/.style={circlenode, fill=white!30, dashed},
    invisiblenode/.style={circle, draw=none, minimum size=0.75cm, inner sep = 0pt},
    clusterinteraction/.style={->, thick, color=black, decoration={markings, mark=at position 0.5 with {\node[draw,circle,fill=black,inner sep=1pt] {};}}, postaction={decorate}},
    repulsioninteraction/.style={->, thick, color=black, decoration={markings, mark=at position 0.5 with {\draw[line width=0.8pt,-] (-0.1,-0.1) -- (0.1,0.1); \draw[line width=0.8pt,-] (-0.1,0.1) -- (0.1,-0.1);}}, postaction={decorate}},
    equalsinteraction/.style={->, thick}
}

\usepackage{physics}
\usepackage{bbm}
\DeclareMathOperator*{\argmin}{argmin}

\newcommand{\indicator}{\mathbbm{1}}

\renewcommand{\ip}[2]{#1\cdot #2}

\newcommand{\NN}{\mathbb{N}}
\newcommand{\ZZ}{\mathbb{Z}}
\newcommand{\RR}{\mathbb{R}}

\newcommand{\EE}[1]{\mathbb{E} \left[ #1 \right]}

\newcommand{\RRd}{{\RR^d}}

\newcommand{\de}{\mathrm{d}}

\newcommand{\cov}[1]{\, {\rm cov}\left( #1 \right) }

\renewcommand{\var}[1]{\, {\rm var}\left( #1 \right) }

\newcommand{\set}[1]{\left\{#1\right\}}

\newcommand{\freq}{{k}}

\newcommand{\momentmeasure}[1]{M_{#1}}
\newcommand{\cumulantmeasure}[1]{C_{#1}}
\newcommand{\reducedmeasure}[2]{\Breve{#1}_{#2}}
\newcommand{\reducedmomentmeasure}[1]{\reducedmeasure{M}{#1}}
\newcommand{\reducedmomentdens}[1]{\reducedmeasure{m}{#1}}
\newcommand{\reducedcumulantmeasure}[1]{\reducedmeasure{C}{#1}}
\newcommand{\reducedcumulantdens}[1]{\reducedmeasure{c}{#1}}

\newcommand{\Kfunc}[1]{{K}_{#1}}
\newcommand{\Lfunc}[1]{{L}_{#1}}
\newcommand{\paircf}[1]{g_{#1}}
\newcommand{\Cfunc}[1]{C_{#1}}
\newcommand{\Cfuncest}[1]{\widehat{C}_{#1}}

\newcommand{\sdf}[1]{f_{#1}}

\newcommand{\intensitybase}{\lambda}
\newcommand{\intensity}[1]{\intensitybase_{#1}}
\newcommand{\intensityest}[1]{\hat{\intensitybase}_{#1}}

\newcommand{\region}{\mathcal{R}}

\newcommand{\conj}[1]{\overline{#1}}

\newcommand{\dft}[2]{J_{#1;#2}}
\newcommand{\pgram}[2]{I_{#1;#2}}
\newcommand{\mtpgram}[1]{\hat{f}_{#1}}

\newcommand{\borel}[1]{\mathcal{B}(#1)}

\newcommand{\sphere}[1]{S^{#1}}
\newcommand{\zerosphere}[1]{\sphere{#1}_0}
\newcommand{\UU}{\mathbb{U}}
\newcommand{\UUd}{{\UU^d}}
\newcommand{\totallyfinite}[1]{\mathcal{S}_{#1}}
\newcommand{\Poisson}{\mathrm{Poisson}}
\newcommand{\points}{\Phi}
\newcommand{\allprocesses}{V}

\newcommand{\pred}{\bullet}

\newcommand{\clusternumber}[1]{\Poisson(\clusteraverage{#1})}
\newcommand{\clusteraverage}[1]{\eta_{#1}}
\newcommand{\clusteroffset}[1]{P_{#1}}
\newcommand{\clusterdens}[1]{p_{#1}}
\newcommand{\clusterchar}[1]{\phi_{#1}}

\newcommand{\leb}[1]{\ell\left(#1\right)}

\newcommand{\sdftopartial}[1]{F_{#1}}

\newcommand{\Predprocess}[1]{\Lambda_{#1}}
\newcommand{\predprocess}[1]{\lambda_{#1}}
\newcommand{\predkernel}[1]{w_{#1}}
\newcommand{\predkernelft}[1]{\tilde{w}_{#1}}
\newcommand{\Predkernel}[1]{W_{#1}}
\newcommand{\predkernelest}[1]{\widehat{w}_{#1}}

\newcommand{\centeredresidualprocess}[1]{\varepsilon^0_{#1}}
\newcommand{\residualprocess}[1]{\varepsilon_{#1}}

\newcommand{\oneFtwo}[3]{{{}_1F_2} \left(#1;#2;#3\right)}

\newcommand{\wavenumberspacing}[1]{\freq^{\Delta}_{#1}}
\newcommand{\wavenumbermax}[1]{\freq^{\text{max}}_{#1}}
\newcommand{\regionlength}[1]{L_{#1}}
\newcommand{\waveradii}{\kappa}
\newcommand{\waveradiispacing}{\waveradii^{\Delta}}
\newcommand{\waveradiimax}{\waveradii^{\text{max}}}
\newcommand{\waveradiigrid}{\overline{\Omega}}
\newcommand{\shellwavenumbers}[1]{\Omega_{#1}}

\newcommand{\besselj}[1]{\mathcal{J}_{#1}}

\usepackage[capitalise]{cleveref}
\crefname{assumption}{Assumption}{assumptions}
\Crefname{assumption}{Assumption}{Assumptions}
\crefname{proposition}{Proposition}{propositions}
\Crefname{proposition}{Proposition}{Propositions}
\crefname{theorem}{Theorem}{theorems}
\Crefname{theorem}{Theorem}{Theorems}

\usepackage{xr}

\AtBeginEnvironment{appendices}{\crefalias{section}{appendix}}

\usepackage{etoolbox}
\usepackage{booktabs}
\usepackage{pdflscape}
\usepackage{mathtools}

\newcommand{\diagramkey}[1]{%
    \node[below=0.5cm of #1] (keycenter) {}; 

    \node[font=\scriptsize] at (keycenter) (labelequal) {Copying};
    \node[below=0.1cm of labelequal, xshift=-0.5cm] (keyequalstart) {};
    \node[right=0.8cm of keyequalstart] (keyequalend) {};
    \draw[equalsinteraction] (keyequalstart) -- (keyequalend);

    \node[font=\scriptsize, right=0.5cm of labelequal] (labelcluster) {Clustering};
    \node[below=0.1cm of labelcluster, xshift=-0.5cm] (keyclusterstart) {};
    \node[right=0.8cm of keyclusterstart] (keyclusterend) {};
    \draw[clusterinteraction] (keyclusterstart) -- (keyclusterend);

    \node[font=\scriptsize, right=0.5cm of labelcluster] (labelrepulsion) {Repulsion};
    \node[below=0.1cm of labelrepulsion, xshift=-0.5cm] (keyrepulsionstart) {};
    \node[right=0.8cm of keyrepulsionstart] (keyrepulsionend) {};
    \draw[repulsioninteraction] (keyrepulsionstart) -- (keyrepulsionend);

    \node[font=\scriptsize, left=0.5cm of labelequal] (labellatent) {Latent};
    \node[latentnode, minimum size = 0.3cm, below=0.2cm of labellatent] (keylatent) {};
    \node[font=\scriptsize, left=0.5cm of labellatent] (labelrecord) {Recorded};
    \node[recordednode, minimum size = 0.3cm, below=0.2cm of labelrecord] (keyrecord) {};
}

\title{The partial $K$ function}
\renewcommand{\shorttitle}{The partial $K$ function}

\newbox{\orcid}\sbox{\orcid}{\includegraphics[scale=0.06]{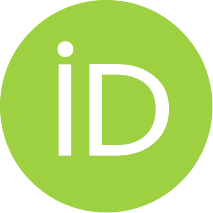}} 
\author[1]{\href{https://orcid.org/0000-0002-8808-4821}{\usebox{\orcid}}\hspace{1mm}{J. P. GRAINGER}\thanks{\texttt{jake.grainger@epfl.ch}}}
\author[2]{\href{https://orcid.org/0000-0002-1343-8058}{\usebox{\orcid}}\hspace{1mm}{T. A. RAJALA}\thanks{\texttt{tuomas.rajala@luke.fi}}}
\author[3]{\href{https://orcid.org/0000-0002-4830-8966}{\usebox{\orcid}}\hspace{1mm}{D. J. MURRELL}\thanks{\texttt{d.murrell@ucl.ac.uk}}}
\author[1]{\href{https://orcid.org/0000-0003-0061-227X}{\usebox{\orcid}}\hspace{1mm}{ S. C. OLHEDE}\thanks{\texttt{sofia.olhede@epfl.ch}}}
\affil[1]{Institute of Mathematics, \'Ecole Polytechnique F\'ed\'erale de Lausanne, Station 8, 1015 Lausanne, Switzerland }
\affil[2]{Natural Resources Institute Finland, 00790 Helsinki, Finland }
\affil[3]{Research Department of Genetics, Evolution and Environment, Centre for Biodiversity and Environment Research, University College London, UK }

\hypersetup{
pdftitle={The partial $K$ function},
pdfsubject={stat.ME},
pdfauthor={J. P. GRAINGER, T. A. RAJALA, D. J. MURRELL,  S. C. OLHEDE},
pdfkeywords={},
}

\begin{document}

\maketitle
\begin{abstract}
The $K$ function and its related statistics have been an enduring tool in the analysis of spatial point processes, providing an easy to compute and interpret summary statistic for characterising the interactions between points of one type, or between two different types of points.
    In this paper, we introduce a partial $K$ function, enabling us to account for some of the effects of the other point types when analysing point-point interactions.
    The partial $K$ function we introduce reduces to the usual $K$ function when the other points are independent of the points of interest and has a similar interpretation.
    Using examples, we demonstrate how the partial $K$ function can unpick dependence between point types that would otherwise be hidden in the usual $K$ function.
    We also discuss important bias correction steps and hyperparameter selection.
    In addition, we introduce an extension to account for other spatial covariates, and demonstrate the methodology on the Lansing Woods dataset. 

\end{abstract}

\section{Introduction}

Point patterns often arise from systems with complicated interactions between different types of points. 
For example, we might record the locations of individuals of multiple species of tree in a forest, all of which may influence each other.
In such multivariate systems, apparent associations between two components may be due to the influence of other components in the system, and this can lead to spurious associations if not accounted for.
Partial statistics are a powerful tool for separating such complicated dependence structures, and have been widely used for multivariate random variables \citep{baba2004partial}, time series \citep{dahlhaus2000graphical} and random fields \citep{guinness2014multivariate}.
For spatial point processes, however, existing approaches have been purely Fourier domain constructions \citep[for example]{eckardt2019analysing}, which can be difficult to interpret and compare with the standard tools used in practice.

In spatial point process analysis, functional summary statistics, such as the long established $K$ function \citep{ripley1977modelling}, provide a useful way to explore interactions between points in a point pattern.
These existing approaches perform well when we are interested in one or two types of points, and have been modified to account for covariates and long scale inhomogeneity \citep{baddeley2000non}. 
However, if there are more than two types of points present we may want to also account for the influence of these additional types of points.
In particular, the influence of a third type of point may cause spurious clustering or repulsion between the two types of interest.
For example, if points of type $X$ and type $Y$ both cluster around points of type $Z$, then the usual $K$ function may suggest clustering between $X$ and $Y$ even if there is no direct interaction between them.
This is the limitation we address here.

In what follows, we introduce partial versions of the standard second-order statistics used in the point process literature, in particular, any which can be computed from Ripley's $K$ function. 
This complements the basic $K$ function and the intensity-reweighted $K$ function \citep{baddeley2000non} with the new \emph{partial $K$ function}.
The partial $K$ function is particularly useful if we suspect that interactions between or within types could be due to the influence of other types of points, but we do not have a good model \textit{a priori} to describe such dependence.
If we have two types of points of interest, $X$ and $Y$, Ripley's $K$ function is informally
\begin{equation*}
    K_{XY}(r) = \frac{1}{\intensity{X}}\EE{\begin{array}{l}\text{number of extra points of type $X$ within} \\\text{a distance $r$ of a typical point of type $Y$}\end{array}}, 
\end{equation*}
where $\intensity{X}$ is the average number of type $X$ points per unit area \citep{dixon2013ripley}.
Note that the term ``extra'' is only necessary when $X=Y$ to exclude the typical point itself from the count (provided the processes are jointly simple).
In essence, this relates to counting the number of points of one type within a distance $r$ of points of another type.
Such a statistic can be computed easily, has a straightforward interpretation, and is easy to visualise (being a function of distance).

If we have additional points of type(s) $Z=(Z_1,\ldots,Z_q)$, which we call covariate processes, the partial $K$ function is informally
\begin{equation*}
    K_{XY\pred Z}(r) = K_{XY}(r) - \frac{1}{\intensity{X}} \EE{\begin{array}{l}\text{excess intensity of points of type $X$}\\\text{linearly predicted from points of type $Z$}\\\text{within distance $r$ of a typical point of type $Y$}\end{array}}.
\end{equation*}
In other words, we take the usual $K$ function and remove from it the clustering/repulsion which could be explained by linear prediction from $Z$, the covariate processes.
The reason we do this linearly is because we can easily construct non-parametric estimators, without needing to specify a model \textit{a priori}.
This linear construction does have limitations: the resulting statistic is partial and not conditional \citep{baba2004partial}.
This means one cannot interpret the partial $K$ function as detecting conditional orthogonality.
However, the partial $K$ function still provides a powerful tool for unpicking complicated dependence structures which may be present in the observed point patterns.

The main contributions of this paper are as follows.
First, we introduce a family of partial functional statistics, which reduce to the classical point process equivalents when the covariate processes are independent of the processes of interest.
Second, we provide reliable estimates for these quantities, including developing important debiasing steps which we show can have significant impact on the quality of the estimators.
Third, we give detailed examples illustrating how these statistics work in a variety of cases designed to correspond to phenomena that are of interest in practice.
Finally, we provide a clear interpretation of our novel statistics that can be thought of as taking the existing statistic and making an adjustment to account for the covariate processes, thus making it easier for practitioners to adopt.
Therefore, this paper provides a new angle to analyse point pattern data, enabling us to extract new insights into the processes of interest.

\section{Background}
\label{sec:background}

\subsection{Point processes}
Point processes are a type of stochastic process that model the random locations of points in some space, e.g.\ $\RR$ for temporal point processes or $\RR^2$ for spatial point processes.
The material in this section can be found in standard textbooks on point processes, such as \cite{moller2003statistical}, \cite{daley2003introduction} and \cite{illian2008statistical}.
Formally, a \emph{point process} is a random measure $N$ on $\RR^d$ taking values in the non-negative integers $\NN_0$.
For a Borel set $B\in\borel{\RR^d}$, $N(B)$ is the random number of points in $B$.
The \emph{support} of $N$ is the set of point locations in $\RR^d$, which we denote by $\points$.
We say that a point process is \emph{simple} if no two points occur in the same location almost surely (for all $x\in\RRd$, $N(\set{x})\in\set{0,1}$ a.s.), which is usually a natural assumption when modelling distinct locations.
Such a framework provides a rigorous formalism to develop models for applications where point locations are the primary object of interest.

Multitype point processes generalize point processes to settings with multiple types of points. 
For example, say that we have three types, $X$, $Y$ and $Z$, then we can define multiple point processes $N_X$, $N_Y$ and $N_Z$ to represent them, with supports $\points_X$, $\points_Y$ and $\points_Z$, respectively.
The \emph{ground process} is the point process which contains all points of all types, and is defined as $N = N_X + N_Y + N_Z$.
We say the the processes are jointly simple if the ground process is simple, meaning that no two points of any type occur in the same location almost surely.
The labels $X$, $Y$ and $Z$ are arbitrary, but will often play the roles of processes of interest and covariate processes, as described in the introduction.
Therefore we will assume that we have a system of $P$ point processes, $N_1,\ldots,N_P$.
For convenience, we will refer to the collection of all indices as $\allprocesses=\set{1,\ldots,P}$.
Then when referring to specific processes, we tend to use the labels $X,Y\in \allprocesses$ and $Z=\allprocesses\setminus\set{X,Y}$.
Such multitype point processes enable us to model interactions between different types of points, which is important for understanding the underlying structure of the system.

Throughout this work, we assume the system of point processes is jointly \emph{homogeneous}, meaning the joint distribution across process types is invariant under spatial translations \citep[see][for a formal definition]{daley2007introduction}.
Say that we have three types of processes $X$, $Y$ and $Z$.
As previously noted, the $K$ function $K_{XY}$ may be misleading if there is dependence between $X$ and $Y$ induced by $Z$, and this means $K_{XY}$ not only requires that the spatial joint distribution points of types $X$ and $Y$ is homogeneous, but also conditionally on the points of type $Z$.
Our partial $K$ function relaxes this assumption to allow for certain forms of conditional mean inhomogeneity.
We require only that the system of all $P$ processes is jointly homogeneous, not that individual processes are conditionally homogeneous given others.
This is a weaker requirement because processes can exhibit apparent spatial trends due to dependence on other processes, even if the system as a whole is homogeneous.

\subsection{Summary statistics for point processes}
There are many different summary statistics for point processes, see e.g. \cite{illian2008statistical} for an overview.
For clarity of exposition, we will focus on point processes in planar space (i.e. $d=2$); the extension to general dimensions is straightforward, see \cref{app:generaldim}.
The first statistic of interest to us is the \emph{intensity} of a point process.
In general, the intensity measure is $\Lambda_X(B) = \EE{N_X(B)}$ for a Borel set $B\in\borel{\RR^2}$. 
When the process is homogeneous, the intensity measure has a constant density, $\intensity{X}$, which we will call the intensity, and it can be interpreted as the average number of points per unit area.

To introduce the remaining statistics, we first define the second-order moments of point processes. The \emph{second-order moment measure} between two processes $X$ and $Y$ is given by $\momentmeasure{XY}(A\times B) = \EE{N_X(A) N_Y(B)}$ for sets $A,B\in\borel{\RR^2}$.
For homogeneous processes, this measure depends only on the relative position of $A$ and $B$, similar to how the autocovariance of a stationary time series depends only on the lag. In this case, we can define the \emph{second-order reduced moment measure}, $\reducedmomentmeasure{XY}$, which for a bounded set $B\in\borel{\RR^2}$ is
\begin{align}
    \reducedmomentmeasure{XY}(B) & = \EE{\int_{\UU^2} N_X(B+y) N_Y(\de y)},
\end{align}
where $\UU^2=[0,1]^2$ is the unit square \citep{daley2003introduction}. This reduced measure captures the expected interaction between points of type $X$ and $Y$ at a given spatial lag.
In order to define the statistics in this paper, we require that the processes we consider satisfy the following assumption.
\begin{assumption}[Second-order homogeneous and simple]\label{assumption:secondorder}
    We require the following three conditions:
    \begin{enumerate}
        \item The processes are jointly homogeneous and jointly simple.
        \item For all $X \in \allprocesses$, $0<\intensity{X} < \infty$.
        \item For all $X,Y\in \allprocesses$, the second-order moment measure $\momentmeasure{XY}$ exists.
    \end{enumerate}
\end{assumption}
This guarantees a number of useful properties, such as translation-boundedness of the second-order reduced moment measure \citep[Proposition 8.3.I]{daley2003introduction}, which ensures the finiteness of the $K$ function and related statistics.

Given \cref{assumption:secondorder}, Ripley's \emph{$K$ function} is defined as
\begin{align}
    \Kfunc{XY}(r) & = \intensity{X}^{-1} \intensity{Y}^{-1}\reducedmomentmeasure{XY}\left(r\zerosphere{2}\right), & r \geq 0,
\end{align}
where $\zerosphere{2} = \sphere{2} \setminus \set{0}$ and $\sphere{2}$ is the unit ball in $\RR^2$.
If $X=Y$ then this is called the $K$ function of $X$, otherwise it is called the cross $K$ function between $X$ and $Y$.
The $K$ function can be interpreted as the expected number of points of type $X$ within a distance $r$ of a typical point of type $Y$, divided by the expected number of points of type $X$ in a unit area (provided the processes are homogeneous).

Typically, the $K$ function is transformed so that it is easier to interpret visually.
A common transformation is the \emph{$L$ function} \citep{besag1977contribution}, which is
\begin{align}
    \Lfunc{XY}(r) & = \sqrt{\frac{K_{XY}(r)}{\pi}}, & r \geq 0.\label{eq:k2l}
\end{align}
The benefits of this transformation are that $L_{XY}(r) = r$ if the points of type $X$ are independent of points of type $Y$ (or if the process is Poisson when $X=Y$), and the transformation is variance stabilising \citep{besag1977contribution}.
Another common related statistic is the \emph{pair correlation function}, which is defined as
\begin{align}
    \paircf{XY}(r) & = \frac{\Kfunc{XY}'(r)}{2\pi r}, & r > 0.
\end{align}
The pair correlation function is one when the points of type $X$ are independent of points of type $Y$ (or when $X=Y$ is a Poisson process).

A useful statistic for our purposes, which we will call the \emph{$C$ function}, is for $r \geq 0$, $\Cfunc{XY}(r) = \reducedcumulantmeasure{XY}(r\zerosphere{2})$, where $\reducedcumulantmeasure{XY}(B) = \reducedmomentmeasure{XY}(B) - \intensity{X} \intensity{Y} \leb{B}$ is the reduced covariance measure (for $B\in\borel{\RR^2}$).
The $C$ function is easily related to the $K$ function and pair correlation function by
\begin{align}
    \Kfunc{XY}(r) & = \frac{\Cfunc{XY}(r)}{\intensity{X} \intensity{Y}} + \pi r^2, & r \geq 0\label{eq:c2k}\\
    \paircf{XY}(r) & = \frac{\Cfunc{XY}'(r)}{2 \pi r \intensity{X} \intensity{Y}} + 1, & r > 0.\label{eq:c2g} 
\end{align}
When the points of type $X$ are independent of points of type $Y$, then $\Cfunc{XY}(r) = 0$ for all $r\geq 0$.
The spectral density function is given by
\begin{align}
    \sdf{XY}(\freq) & = \int_{\RR^2} e^{-2\pi i \ip{u}{\freq}} \reducedcumulantmeasure{XY}(\de u), & \freq \in \RR^2,
\end{align}
provided that $\reducedcumulantmeasure{XY}$ is totally finite \citep{daley2003introduction}.
For partial statistics, it is easiest to work with the spectral density function, but for inference spatial domain is more intuitive, so transforming from the spectral density function to spatial statistics is a useful step.
This requires the following additional integrability condition.
\begin{assumption}[Atom-corrected integrability]\label{assumption:sdfL1}
    For all $X,Y\in\allprocesses$, $\sdf{XY} - \reducedmomentmeasure{XY}(\set{0})$ is integrable.
\end{assumption}
Then we can relate the $C$ function and its derivative to the spectral density function with the following proposition.
\begin{proposition}[Spectral inversion for the $C$ function]\label[proposition]{prop:sdf2c}
    Under \cref{assumption:secondorder,assumption:sdfL1},
    \begin{align}
        \Cfunc{XY}(r)  & = \int_{\RR^2} \left[\sdf{XY}(\freq) - \reducedmomentmeasure{XY}(\set{0})\right] \frac{r}{\norm{\freq}} \besselj{1}(2\pi\norm{\freq}r)\de\freq, \label{eq:sdf2c}           \\
        \Cfunc{XY}'(r) & = \int_{\RR^2} \left[\sdf{XY}(\freq) - \reducedmomentmeasure{XY}(\set{0})\right] 2\pi r \besselj{0}(2\pi\norm{\freq}r) \de\freq,\label{eq:sdf2c:derivative}
    \end{align}
     for all $r\geq 0$, where $\besselj\nu$ is the Bessel function of the first kind of order $\nu$. 
    \begin{proof}
        See \cref{app:spectratocfunction}.
    \end{proof}
\end{proposition}

If the processes $X$ and $Y$ are jointly simple (under \cref{assumption:secondorder}), we have $\reducedmomentmeasure{XX}(\set{0}) = \intensity{X}$ and $\reducedmomentmeasure{XY}(\set{0}) = 0$ when $X\neq Y$.
In practice, \cref{eq:sdf2c:derivative} can be used to compute the pair correlation function, or we can use existing approaches to obtain an estimate from the estimated $K$ function using smoothing splines, either directly from the $K$ function or from some transformation, see \cite{baddeley2005spatstat} for details.

\section{Partial $K$ function}\label{sec:partial}

\subsection{Comparison to existing methods}
Partial statistics in the Fourier domain are computationally efficient linear methods used to account for the influence of covariate processes.
Such methods have been proposed for time series \citep{brillinger1974timeseries, dahlhaus2000graphical}, random fields \citep{guinness2014multivariate}, spatial point processes \citep{eckardt2019analysing,eckardt2019partial,eckardt2021graphical} and combinations of point processes and random fields \citep{grainger2025spectral}. 
In the spatial point process setting, these approaches have been purely Fourier domain constructions. However, in many applications where spatial point patterns are recorded, the wavenumber-domain spectrum is unfamiliar to practitioners, and can be fundamentally hard to interpret. 
For example, while ocean waves are naturally understood as combinations of different frequency components (waves), the same wavenumber decomposition is less intuitive for analysing, say, tree locations in a forest. 
In contrast, spatial-domain statistics such as Ripley's $K$ function and the pair correlation function can be seen as giving a ``plant's eye view'' of the community \citep{law2009ecological}. 
Our goal is to solve this problem by developing partial analogues of the standard spatial domain summaries derived from Ripley's $K$ function, with an interpretation on the same spatial scale as the ordinary $K$ and $L$ functions, whilst maintaining the ability of the partial spectra to account for the other point types in the system of processes.

For temporal point processes, \cite{eichler2003partial} proposed a time-domain statistic based on transforming partial spectra back from the wavenumber domain.
However, their approach has significant limitations for spatial applications.
Their statistic is equivalent to a one-dimensional scaled version of the derivative of what we call the $C$ function (see \cref{eq:sdf2c:derivative}), which is not standard in the spatial point process literature, making it difficult to compare with existing methods and hindering adoption by practitioners.

Our proposed statistic resolves these issues by introducing a partial equivalent of the widely-used $K$ function.
The partial $K$ function we develop is a direct generalization that reduces to exactly the usual $K$ function when there is no influence from the covariate processes, enabling practitioners to interpret partial and non-partial statistics on the same scale using familiar methodology.


\subsection{Definition of the partial statistics}\label{sec:partial:definition}

We will usually refer to the process ($X$) or pair of processes ($X$ and $Y$) which we are considering the interactions between as the \emph{processes of interest}.
The process(es) which we are accounting for ($Z$), we call the \emph{covariate processes}. Recall, the collection of all processes we denote by $V$, i.e. $X,Y\in V$ and $Z= V\setminus\{X,Y\}$.

In the time series setting, \cite{brillinger1974timeseries} proposed the partial spectra between time series $X$ and $Y$ accounting for other processes $Z=(Z_1,\ldots,Z_q)^\top$ as the spectra of residual processes obtained by linearly predicting $X$ and $Y$ from $Z$ and taking the difference between the original processes and their predictions.
For example, the prediction process takes the form
\begin{align}
    \hat{X}_{t\pred Z} & = \sum_{s=-\infty}^\infty h_{s,X\pred Z} Z_{t-s}, & t \in \ZZ,
    \nonumber
\end{align}
where $h_{s,X\pred Z}$ is a sequence of row vectors of filter coefficients, and the residuals are $\varepsilon_{t, X\pred Z} = X_t - \hat{X}_{t\pred Z}$ (assuming the time series are mean-zero and second-order stationary).

However, in the point process setting (or random measures more generally), this is much more difficult to define.
The first difficulty is that the corresponding notion of prediction process is more complicated, as the prediction process is no longer the convolution of two sequences, but rather should be the convolution of a point process with a signed measure \citep{daley2003introduction}, which is more difficult to work with. 
Secondly, the notion of residual process itself may not be well defined as a random signed measure as, for unbounded sets, both the point process and prediction process may be infinite at the same time \citep[][]{passeggeri2025random}.
Therefore, we will begin by giving definitions for the partial spectral density function and partial $C$ function (and implicitly the partial $K$ function), which exist under mild assumptions.
Then, we will give an interpretation of these partial statistics in terms of residual processes, which holds under stronger assumptions, but nonetheless provides useful motivation.

To define the partial statistics, we will need to make some additional assumptions.
Beginning with the partial spectral density function, we need the following assumption to ensure that it is well defined.
\begin{assumption}[Invertibility of the spectral density matrix]\label{assumption:invertiblespectra}
    The spectral density matrix function at wavenumber $\freq$, $\sdf{VV}(\freq)$, is invertible for all wavenumbers $\freq \in \RR^2$.
\end{assumption}
Provided that \cref{assumption:invertiblespectra} holds, we define the partial spectral density function between $X$ and $Y$ accounting for $Z$ as
\begin{equation}
    \sdf{XY\pred Z}(\freq) 
    = \sdf{XY}(\freq) - \sdf{XZ}(\freq) \sdf{ZZ}(\freq)^{-1} \sdf{ZY}(\freq),\label{eq:def:partialspectra}
\end{equation}
which mirrors the result in the time series case \citep[Theorem 8.3.1]{brillinger1974timeseries}.
Here, when $Z=(Z_1,\ldots,Z_q)^\top$ contains multiple processes, $\sdf{XZ}(\freq)$ is the $1\times q$ row vector of cross-spectra between $X$ and the components of $Z$, $\sdf{ZZ}(\freq)$ is the $q\times q$ spectral density matrix of $Z$, and $\sdf{ZY}(\freq)$ is the corresponding $q\times 1$ column vector.
Next, motivated by \cref{prop:sdf2c}, we make the following integrability assumption to ensure that the partial $C$ function is well defined.
\begin{assumption}[Partial atom-corrected integrability]\label{assumption:partialcL1}
    For all $X,Y\in\allprocesses$, $\sdf{XY\pred Z}(\freq)-\reducedmomentmeasure{XY}(\set{0})$ is integrable.
\end{assumption}
That $\reducedmomentmeasure{XY}(\set{0})$ plays the role of $\reducedmomentmeasure{XY\pred Z}(\set{0})$ is due to the joint simplicity part of \cref{assumption:secondorder} and \cref{assumption:sdfL1}.
If the latter assumption does not hold, then it is no longer true that $\reducedmomentmeasure{XY}(\set{0})=\reducedmomentmeasure{XY\pred Z}(\set{0})$, see e.g. the first example in \cref{sec:partial:examples}.

We define the partial $C$ function between $X$ and $Y$ accounting for $Z$ as
\begin{align}
    \Cfunc{XY\pred Z}(r)  & = \int_{\RR^2} \left[\sdf{XY\pred Z}(\freq) - \reducedmomentmeasure{XY}(\set{0})\right] \frac{r}{\norm{\freq}} \besselj{1}(2\pi\norm{\freq}r)\de\freq, & r \geq 0.\label{eq:def:partialc}
\end{align}
Of course, given our assumptions, the derivative of this partial $C$ function is
\begin{align}    
    \Cfunc{XY\pred Z}'(r) & = \int_{\RR^2} \left[\sdf{XY\pred Z}(\freq) - \reducedmomentmeasure{XY}(\set{0})\right] 2\pi r \besselj{0}(2\pi\norm{\freq}r) \de\freq, & r \geq 0.
\end{align}
The partial $K$ function can then be defined from the partial $C$ function using the same relationship as in \cref{eq:c2k}, replacing $C_{XY}$ with $C_{XY\pred Z}$, but retaining the same intensities $\intensity{X}$ and $\intensity{Y}$. Keeping these intensities corresponds to a certain choice of centring for the residual processes, but we will see that this provides a natural interpretation of the partial statistics in terms of residual processes in \cref{sec:partial:interpretation}.
Similarly, the partial pair correlation function using \cref{eq:c2g}, and partial $L$ function using \cref{eq:k2l}.
If $X$ and $Y$ are independent of $Z$, then $\Kfunc{XY\pred Z}(r) = \Kfunc{XY}(r)$, so the partial $K$ function reduces to the usual $K$ function, as we would expect (both when $X=Y$ and when $X\neq Y$).

The partial $K$ function is not guaranteed to be positive nor increasing.
If the partial $K$ function is negative, then the partial $L$ function is not real-valued, as the $L$ function square roots the $K$ function.
Therefore, we will replace the $L$ function with a signed version, so that instead
\begin{align}
    \Lfunc{XY\pred Z}(r) & = \mathrm{sgn}(K_{XY\pred Z}(r))\sqrt{\frac{\abs{K_{XY\pred Z}(r)}}{\pi}}, & r\geq 0,
\end{align}
where $\mathrm{sgn}$ is the sign function.
For the $K$ function, the signed $L$ function is the usual $L$ function, as the $K$ function is always non-negative.

\subsection{Residual process interpretation}\label{sec:partial:interpretation}

Now in order to develop an interpretation of the partial statistics in terms of residual processes, we need to make some stronger assumptions.
We will start with the prediction process.
For convenience, let $N_X^0(B) = N_X(B)-\intensity{X}\leb{B}$ be the centred count ``measure'' for bounded $B\in\borel{\RR^2}$.\footnote{Strictly speaking this is a set function not a measure.} Integrals with respect to $N_X^0$ will be taken to mean the difference between the integral with respect to $N_X$ and the integral with respect to $\intensity{X}\ell$.
We will now define prediction of $N_X^0$ from the collection of centred point processes $N_Z^0 = (N_{Z_1}^0,\ldots,N_{Z_q}^0)^\top$ as the sum of convolutions of $N_{Z_j}^0$ with a collection of signed measures $\Predkernel{}=(\Predkernel{1},\ldots,\Predkernel{q})$, in particular our prediction of $N_X^0(B)$ for bounded $B\in\borel{\RR^2}$ is of the form
\begin{align}
    \sum_{j=1}^q \int_{\RR^2} \Predkernel{j}(B-x) N_{Z_j}^0(\de x).\label{eq:prediction:general}
\end{align}
Now we need to ensure that \cref{eq:prediction:general} is well defined for bounded Borel sets.
A sufficient condition to ensure that each integral is well defined is that for each $j$, $\Predkernel{j}$ is a totally finite signed measure, and $\lambda_{Z_j}$ is finite, which ensures that the integral is finite almost surely for bounded sets.
In addition, this assumption also ensures that $\Predkernel{}$ also has a Fourier transform, $\predkernelft{}$, which exists, and is bounded and continuous.

To define the best linear prediction, we seek the row vector of signed measures $\Predkernel{}$ minimizing the mean squared prediction error on a bounded test set. Unlike in the time series setting, where prediction is evaluated at a single time point, here we use the reference set $\UU^2$. This choice is essentially arbitrary, since it suffices to identify the Fourier transform of the optimal prediction kernel almost everywhere, and the Fourier transform of $\indicator_{\UU^2}$ is non-zero almost everywhere.

\begin{assumption}[Existence of a totally finite prediction kernel]\label{assumption:predictionkernel}
    Assume that
    \begin{align*}
        \predkernelft{X\pred Z}=\sdf{XZ}\sdf{ZZ}^{-1}
    \end{align*}
    is the Fourier transform of a $1\times q$ row vector of totally finite signed measures
    \begin{align*}
        \Predkernel{X\pred Z}=(\Predkernel{X\pred Z,1},\ldots,\Predkernel{X\pred Z,q}).
    \end{align*}
\end{assumption}

\begin{proposition}[Best linear prediction]\label[proposition]{prop:bestprediction}
    Let $Z=(Z_1,\ldots,Z_q)^\top$ be a vector of covariate processes.
    Under \cref{assumption:secondorder,assumption:sdfL1,assumption:invertiblespectra,assumption:partialcL1,assumption:predictionkernel}, let $\Predkernel{X\pred Z}$ be as in \cref{assumption:predictionkernel}. Then
    \begin{align*}
        \Predkernel{X\pred Z} \in \argmin_{\Predkernel{} \in \left(\totallyfinite{\RR^2}\right)^{1\times q}} \EE{\left(N_X^0(\UU^2) - \sum_{j=1}^q \int_{\RR^2} \Predkernel{j}(\UU^2-x) N_{Z_j}^0(\de x) \right)^{\hspace{-0.25em}2}\;},
    \end{align*}
    where $\left(\totallyfinite{\RR^2}\right)^{1\times q}$ is the set of length $q$ row vectors of totally finite signed Borel measures on $\RR^2$.
     \begin{proof}
        See \cref{app:bestprediction}.
    \end{proof}
\end{proposition}
Intuitively, the formula $\predkernelft{X\pred Z}=\sdf{XZ}\sdf{ZZ}^{-1}$ is the spectral analogue of the ordinary regression coefficient $\hat\beta=\cov{X,Z}/\var{Z}$, applied wavenumber-by-wavenumber.
Therefore the best linear prediction of $N_X^0$ from $N_Z^0$ is given by $\Predprocess{X\pred Z}^0(B) = \sum_{j=1}^q \int_{\RR^2} \Predkernel{X\pred Z,j}(B-x) N_{Z_j}^0(\de x)$ for bounded $B\in\borel{\RR^2}$.
Notice here $\Predkernel{X\pred Z,j}$ is the $j$-th component of $\Predkernel{X\pred Z}$, which is a row vector of signed measures, and is not in general equal to $\Predkernel{X\pred {Z_j}}$, which would be the prediction kernel obtained by only using $Z_j$ as a covariate process.

The residual process is then defined for bounded $B\in\borel{\RR^2}$ as $\residualprocess{X\pred Z}(B) = N_X(B) - \Predprocess{X\pred Z}^0(B)$, so that $\EE{\residualprocess{X\pred Z}(B)}=\intensity{X}\leb{B}$.
Unlike in the time series setting, however, the uncentred residual process is the more natural object for the $K$ function interpretation, because it retains the baseline intensity $\intensity{X}$, which results in an interpretation of the partial $K$ function as the usual $K$ function minus an adjustment term.
So long as $B$ is bounded, $\residualprocess{X\pred Z}(B)$ is well defined, as we have already established that $N_X(B)$ and $\Predprocess{X\pred Z}^0(B)$ are finite almost surely.

All that remains is to show that the residual process is connected to the partial $C$ function as we would expect. Whilst the partial spectral density function is often defined through residual processes, since the predictors are the same for both residual processes it is equivalent to work with one residual process and the other original process. 
This yields a more useful interpretation in the point process setting. In particular, for $r\geq 0$ define

\begin{align}
    \Cfunc{\residualprocess{X\pred Z}, Y}(r) & = \EE{\int_{\UU^2} \residualprocess{X\pred Z}(r\zerosphere{2}+y) N_Y(\de y)} - \intensity{X}\intensity{Y}\leb{r\zerosphere{2}}.
\end{align}

Importantly, this $C$ function is only evaluated for finite $r$, so the integration is over bounded sets, ensuring the residual process is well defined for these sets. We do not require the reduced covariance measure of the residual process, which may or may not exist (see \cite{daley2003introduction}, Chapter 8).
In addition, the $C$ function of this residual process is equivalent to the partial $C$ function constructed in \cref{sec:partial:definition}.

\begin{proposition}[Equivalence of partial and residual $C$ functions]\label[proposition]{prop:partialcov:isresidualcov}
    Under \cref{assumption:secondorder,assumption:sdfL1,assumption:invertiblespectra,assumption:partialcL1,assumption:predictionkernel}, we have for all $r\geq 0$ that
         $\Cfunc{XY\pred Z}(r) = \Cfunc{\residualprocess{X\pred Z}, Y}(r).$
     \begin{proof}
        See \cref{sec:partialcov:isresidualcov}.
    \end{proof}
\end{proposition}

The usual $K$ function has a clean interpretation in terms of typical points, so that $\lambda_X\Kfunc{XY}(r)$ is the expected number of points of type $X$ within a distance $r$ of a typical point of type $Y$.
This interpretation arises from Palm theory \citep[see e.g.][Section C.2]{moller2003statistical}, which allows us to think about the distribution of a point process conditional on a point being at a certain location \citep[][Section 4.1]{illian2008statistical}.
In fact, we can construct a similar interpretation for the partial $K$ function.

In particular, writing $\residualprocess{X\pred Z}(B)=N_X(B)-\Predprocess{X\pred Z}^0(B)$ for the residual process, we have from \cref{prop:partialcov:isresidualcov} that
\begin{align}
    \intensity{X}\Kfunc{XY\pred Z}(r)
     & = \intensity{Y}^{-1} \left(\EE{\int_{\UU^2} \residualprocess{X\pred Z}(r\zerosphere{2}+y) N_Y(\de y)} - \intensity{X}\intensity{Y}\leb{r\zerosphere{2}}\right) \nonumber \\
     & = \mathbb{E}_0^Y\left[\residualprocess{X\pred Z}(y+r\zerosphere{2})\right] \nonumber \\
     & =\lambda_X K_{XY}(r)-  \mathbb{E}_0^Y\left[\Predprocess{X\pred Z}^0(r\zerosphere{2})\right]\label{eq:partialKpalm}
\end{align}
where $\mathbb{E}_0^Y$ denotes the expectation conditional on there being a point of type $Y$ at the origin.
Therefore, we can interpret $\intensity{X}\Kfunc{XY\pred Z}(r)$ as the expected number of points of type $X$ within a distance $r$ of a typical point of type $Y$ minus the $Z$-based linear prediction of the excess of $X$ within a distance $r$ of a typical point of type $Y$. Here the excess linear prediction, $\Predprocess{X\pred Z}^0$, is the linear prediction minus the intensity $\intensity{X}$, so that it has mean zero and in particular it can be negative. 
However, its conditional mean on a point of type $Y$ being at the origin, the last term in \cref{eq:partialKpalm}, is not necessarily zero.

\subsection{Some theoretical examples}\label{sec:partial:examples}
The main motivation for using signed measures, $\Predkernel{}$, for the prediction kernel is that in some situations we can predict point locations exactly.
For example, say that we have a process $Z$ and we make $X$ by taking all the points in $Z$ and applying the same deterministic shift, $s\neq 0$, to all points, so that $x=z+s$ (or $N_X(A) = N_Z(A-s)$).
Then we obtain $\reducedcumulantmeasure{XZ}(B) = \reducedcumulantmeasure{ZZ}(B-s)$, so $\sdf{XZ}(\freq) = e^{-2\pi i \ip{\freq}{s}}\sdf{ZZ}(\freq)$ and therefore $\predkernelft{X\pred Z}(\freq) = e^{-2\pi i \ip{\freq}{s}}$.
This is the Fourier transform of the Dirac measure centred at $s$, $\delta_s(A)=\indicator_A(s)$. Furthermore, for bounded $A\in\borel{\RR^2}$
\begin{align*}
    \Predprocess{X\pred Z}(A) 
    & = \intensity{X}\leb{A} + \int_{\RR^2} \delta_s(A-z) N_Z^0(\de z) \\
    & = \int_{\RR^2} \delta_z(A-s) N_Z(\de z) = N_Z(A-s) = N_X(A).
\end{align*}
So, as expected, we can perfectly predict the locations of type $X$ from type $Z$ in this setting.
Notice also that in this case (because $s\neq 0$), \cref{assumption:sdfL1} is violated, as $\sdf{XZ}(\freq) = e^{-2\pi i \ip{\freq}{s}}\sdf{ZZ}(\freq)$, and $\reducedmomentmeasure{XZ}(\set{0})=0$, so $\sdf{XZ}(\freq) - \reducedmomentmeasure{XZ}(\set{0})$ is not integrable.
In this case, because we have perfect prediction the centred residual process is almost surely the zero measure, meaning that its reduced moment measure is also the zero measure, so in particular $\reducedmomentmeasure{XX\pred Z}(\set{0})=0$, not $\lambda_X$.
However, such a setting is somewhat pathological, and the integrability assumption does hold for a wide variety of models.

Usually $\Predkernel{X\pred Z}$ will have a density, which we denote by $\predkernel{X\pred Z}$.
In this case, $\Predprocess{X \pred Z}$ also has a density 
\begin{align*}
    \predprocess{X\pred Z}(u)
     & = \intensity{X} + \int_{\RR^2} \predkernel{X\pred Z}(u-z) N_{Z}^0(\de z)                                                 \\
     & = \intensity{X} - \intensity{Z} \int_{\RR^2} \predkernel{X\pred Z}(u) \de u + \sum_{z \in \points_Z} \predkernel{X\pred Z}(u-z),
\end{align*}
for all $u\in\RR^2$.
So we can see $\predprocess{X\pred Z}(u)$ as being a linear prediction of the intensity of $X$ at a point $u$ in space made by putting a kernel on every point of type $Z$ plus some intercept term.
When $X$ arises from a Neyman-Scott process with clusters centred at $Z$, where the displacement of points within each cluster has a distribution $\clusteroffset{X}$, with the number of points per cluster being $\clusternumber{X}$, we have $\Predkernel{X\pred Z} = \clusteraverage{X} \clusteroffset{X}$ (see \cref{prop:pred:kernel}). When $\clusteroffset{X}$ admits a density, $\clusterdens{X}$, then we have $\predprocess{X\pred Z}(u) = \sum_{z \in \points_Z} \predkernel{X\pred Z}(u-z) = \clusteraverage{X} \sum_{z \in \points_Z}  \clusterdens{X}(u-z)$, for all $u\in\RR^2$.
When $X$ and $Z$ are independent, $\sdf{XZ}=0$, so $\predprocess{X\pred Z}(u)  = \intensity{X}$ for all $u\in\RR^2$.

\subsection{Limitations}
It is important to be aware of the limitations of any statistic.
For the statistics proposed here, the first limitation is that we are considering only second-order properties of the processes.
In the case of the usual $K$ function, \cite{baddeley1984cautionary} provided a nice example of a univariate process whose $K$ function is identical to that of a Poisson process, but which is certainly not a Poisson process.
This limitation extends also to the partial statistics proposed here.

For partial statistics, there is an additional, stronger, limitation.
For random vectors, \cite{baba2004partial} characterise the difference between partial and conditional correlation, which are not necessarily equal in the non-Gaussian setting.
Similarly, if the $\Kfunc{XY}(r) \neq \pi r^2$ then $X$ and $Y$ are dependent.
However, it is possible to have $\Kfunc{XY\pred Z}(r) \neq \pi r^2$, even when $X$ and $Y$ are conditionally independent given $Z$.
This is because the partial $K$ function only accounts for linear trends in conditional means.
It has been claimed in the literature that partial statistics can provide tests of conditional orthogonality/independence \citep{eckardt2019analysing}.
Unfortunately, this is not true, as can be seen from the counter example in \cref{app:counterexample}, and it is important that practitioners are aware of this limitation.
Nonetheless, partial statistics can still be useful to account for covariate processes, especially during exploratory analysis.

\section{Examples}
We illustrate the differences between the $K$ function and the partial $K$ function with simulated examples. We begin with bivariate systems and then consider trivariate systems. For illustration, we compare both the $L$ function and partial $L$ function to null confidence intervals for the $L$ function. Whilst analogous results for the pair correlation function are given in \cref{app:simulation:pcf}, estimates of the pair correlation function are less reliable than those for the $K$ function, as the latter is cumulative. The specific parameter choices for each model are given in \cref{app:simulation_details}.

\subsection{Examples for intraprocess statistics}

Consider a simple predator--prey-like system, where a prey species $Y$ is observed and a predator species $X$ may cluster around it. We consider three scenarios: predators that are conditionally independent given the prey, predators that additionally cluster with each other (packs), and predators that exhibit short-range repulsion (solitary). To generate such examples, we take the prey process $Y$ to be Poisson with intensity $\lambda_Y$, and first generate an intermediate process $X_0$ by placing Gaussian clusters around points of type $Y$, with Poisson($\eta_{X_0}$) offspring per parent. In the independent case we set $X=X_0$. In the packs case we generate a second clustered process around $X_0$. In the solitary case we thin $X_0$ using a marked hard-core-type interaction, so that points close to higher-marked neighbours are removed with high probability. 

\begin{figure}[h]
    \centering
    \begin{tikzpicture}[node distance=0.2cm and 0.5cm]
        \node[recordednode] (Y_1) {$Y$};
        \node[latentnode, right=of Y_1] (X_1_0) {$X_0$};
        \node[recordednode, right=of X_1_0] (X_1) {$X$};
        \draw[clusterinteraction] (Y_1) -- (X_1_0);
        \draw[equalsinteraction] (X_1_0) -- (X_1);
        \node[above=of X_1_0] {independent};

        \node[recordednode, right=4cm of Y_1] (Y_2) {$Y$};
        \node[latentnode, right=of Y_2] (X_2_0) {$X_0$};
        \node[recordednode, right=of X_2_0] (X_2) {$X$};
        \draw[clusterinteraction] (Y_2) -- (X_2_0);
        \draw[clusterinteraction] (X_2_0) -- (X_2);
        \node[above=of X_2_0] {packs};

        \node[recordednode, right=4cm of Y_2] (Y_3) {$Y$};
        \node[latentnode, right=of Y_3] (X_3_0) {$X_0$};
        \node[recordednode, right=of X_3_0] (X_3) {$X$};
        \draw[clusterinteraction] (Y_3) -- (X_3_0);
        \draw[repulsioninteraction] (X_3_0) -- (X_3);
        \node[above=of X_3_0] {solitary};

        \diagramkey{X_2_0}
    \end{tikzpicture}
    \label{fig:predator_prey_dag}
    \caption{Schematic of the predator prey system with three different interaction types.}
\end{figure}
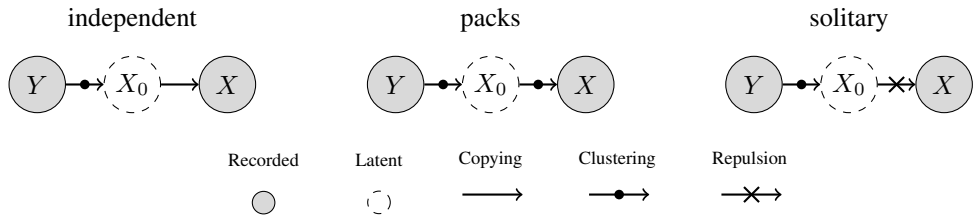

Examples of such processes are shown in \cref{fig:predator_prey_example}. In all three cases, the usual $L$ function indicates clustering. After accounting for the prey process, however, the partial $L$ function separates the three regimes: it is close to the independence line $L(r)=r$ in the independent case, above it in the packs case, and below it in the solitary case. Thus the partial $L$ function recovers the additional within-predator structure that is hidden in the marginal $L$ function.

\begin{figure}[]
    \centering
    \includegraphics[width=1\textwidth]{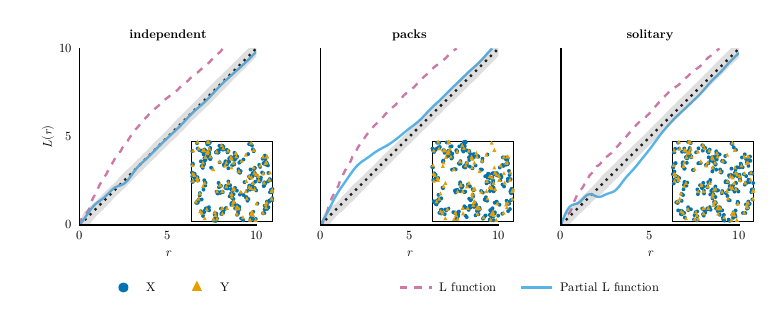}
    \caption{Example of a predator prey system with three different interaction types. The main plots are the $L$ function and partial $L$ function between the predator process ($X$) in the latter case accounting for the prey process ($Y$). The left column shows the first scenario, where the predators do not interact with each other. The middle column shows the second scenario, where the predators hunt in packs. The right column shows the third scenario, where the predators do not like to be near each other. The envelopes are 95\% confidence envelopes for the $L$ function under independence, using the MAD envelopes proposed by \cite{myllymaki2017global}.}
    \label{fig:predator_prey_example}
\end{figure}

\subsection{Examples for interprocess statistics}\label{sec:example:tri}

We next consider a three-process system with processes $X$, $Y$ and $Z$. As before, let $Z$ be Poisson, and generate $Y$ from $Z$ by clustering. We then construct $X$ in one of three ways: independent offspring, where $X$ and $Y$ independently cluster around $Z$; co-operative offspring, where $X$ clusters around $Y$; and antagonistic offspring, where $X$ is first generated from $Z$ and then thinned by $Y$.

Realisations and estimated $L$ and partial $L$ functions are shown in \cref{fig:trivariateexample}. The ordinary $L$ functions largely indicate cross-clustering throughout, whereas the partial $L$ functions recover the underlying structure. In particular, for $X$ versus $Y$ accounting for $Z$, the partial $L$ function shows no interaction, clustering and repulsion in the three scenarios, respectively. For $X$ versus $Z$ accounting for $Y$, it shows clustering, no remaining interaction, and clustering again, matching the way the models are constructed. Finally, for $Y$ versus $Z$ accounting for $X$, clustering remains in all three cases, since $Y$ is always generated from $Z$ in the same way. These examples illustrate how the partial $L$ function can distinguish direct interactions from dependence induced by other processes.
\begin{figure}[]
    \centering
    \begin{tikzpicture}[node distance=0.2cm and 1cm]
        \node[recordednode] (Z_1) {$Z$};
        \node[recordednode] (X_1) [below right=of Z_1] {$X$};
        \node[recordednode] (Y_1) [above right=of Z_1] {$Y$};
        \draw[clusterinteraction] (Z_1) -- (Y_1);
        \draw[clusterinteraction] (Z_1) -- (X_1);
        \node[above=of Y_1, xshift=-0.75cm] {independent offspring};

        \node[recordednode,right=3.5cm of Z_1] (Z_2) {$Z$};
        \node[recordednode] (X_2) [below right=of Z_2] {$X$};
        \node[recordednode] (Y_2) [above right=of Z_2] {$Y$};
        \draw[clusterinteraction] (Z_2) -- (Y_2);
        \draw[clusterinteraction] (Y_2) -- (X_2);
        \node[above=of Y_2, xshift=-0.75cm] {co-operative offspring};

        \node[recordednode,right=3.5cm of Z_2] (Z_3) {$Z$};
        \node[latentnode] (X_3_0) [below right=of Z_3] {$X_0$};
        \node[recordednode] (X_3) [right=of X_3_0] {$X$};
        \node[recordednode] (Y_3) [above right=of Z_3] {$Y$};
        \draw[clusterinteraction] (Z_3) -- (Y_3);
        \draw[clusterinteraction] (Z_3) -- (X_3_0);
        \draw[equalsinteraction] (X_3_0) -- (X_3);
        \draw[repulsioninteraction] (Y_3) -- ($(X_3_0)!0.5!(X_3)$);
        \node[above=of Y_3] {antagonistic offspring};

        \diagramkey{X_2}
    \end{tikzpicture}
    \label{fig:trivariate_dag}
    \caption{Schematic of the trivariate system with three different interaction types.}
\end{figure}
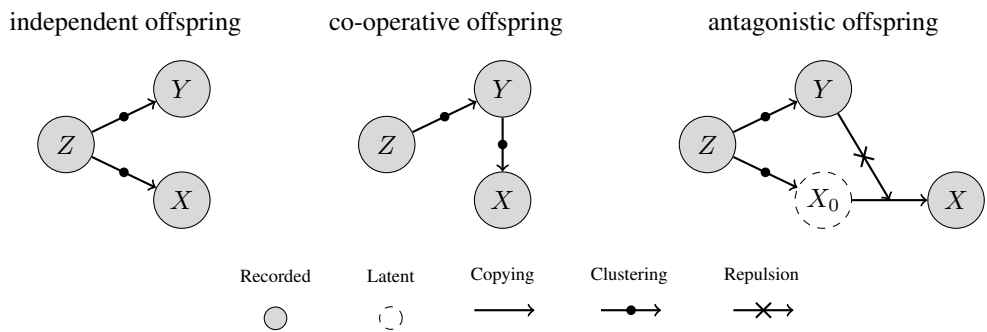
\begin{figure}[]
    \centering
    \includegraphics{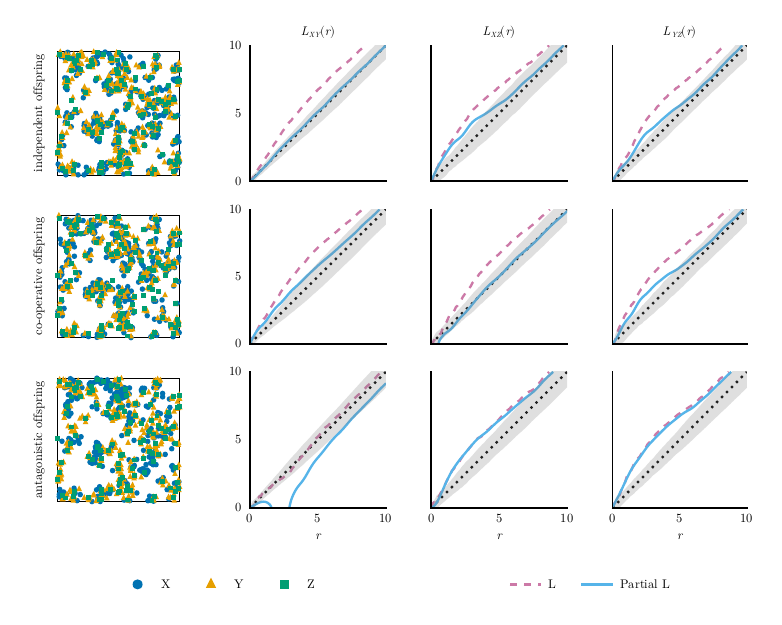}
    \caption{Example of a trivariate system with three different interaction types. The first column shows example processes, the second column shows the $L$ function and partial $L$ function between process $X$ and process $Y$ (possibly accounting for process $Z$), the third column shows interactions between $X$ and $Z$ (accounting for $Y$) and the final column shows interactions between $Y$ and $Z$ (accounting for $Z$).}
    \label{fig:trivariateexample}
\end{figure}



\section{Estimation}\label{sec:estimation}

\subsection{Spectral and partial spectral estimation}
Say we observe data on some bounded observational region $\region$.
To construct estimators for the partial $K$ function, we first need estimators of the spectral density function, and then the partial spectral density function.
We use the multitaper method \citep{thomson1982spectrum, grainger2025spectral} which provides reliable estimators of the spectral density function from observations on arbitrary domains.
Given a family of tapers, $h_1,\ldots,h_M$ supported within $\region$, with $M>P$ tapers, we let
\begin{align}
    \dft{X}{m}(\freq) & = \int_{\RR^2} h_m(x) e^{-2\pi i \ip{x}{\freq}} N_X(\de x) - \hat\lambda_X \int_{\RR^2} h_m(x) e^{-2\pi i \ip{x}{\freq}} \de x, & \freq \in \RR^2. \nonumber
\end{align}
The multitaper periodogram is subsequently defined as
\begin{align}
    \mtpgram{XY}(\freq) & = \frac{1}{M} \sum_{m=1}^M \dft{X}{m}(\freq) \conj{\dft{Y}{m}(\freq)}, & \freq \in \RR^2.
\end{align}
\subsection{Debiasing}
The basic plug-in estimator of the partial spectra is biased even when the spectral density estimates are unbiased. In particular, consider the function mapping spectral density matrices to their partial equivalent,
\begin{align}
    \sdftopartial{XY\pred Z}[f] = \sdf{XY} - \sdf{XZ} \sdf{ZZ}^{-1} \sdf{ZY}
\end{align}
where $\sdf{ZZ}^{-1}$ refers to pointwise matrix inversion (not the inverse of the function $\sdf{ZZ}$).
Then given some regularity conditions, for fixed $M$, $P$, and growing domain, the plug-in estimator satisfies
\begin{align}
    \EE{\sdftopartial{XY\pred Z}[\hat{f}]} = \left(1-\frac{P_Z}{M}\right) \sdftopartial{XY\pred Z}[f] + o(1)
\end{align}
where $P_Z$ is the number of processes in $Z$ (see \cref{theorem:bias:partial}).
Therefore, we use
\begin{align}
    \mtpgram{XY\pred Z} & = \left(\frac{M}{M-P_Z}\right) \sdftopartial{XY\pred Z}[\hat{f}].
\end{align}
This is similar to the bias corrections required for estimating partial coherence \citep{medkour2009graphical}, though it differs slightly because we are directly interested in the partial spectra.
Recall that we already require more tapers than processes, and therefore $M>P_Z$.
When we compute partial $L$ functions, we perform another non-linear transformation, which can warp this bias in unusual ways especially for short distances, which results in bias that is not just a percentage reduction, but which looks like a meaningful feature.
A specific example of this phenomenon is given in \cref{app:simulation:debiasing}.
The debiasing we propose resolves this problem, and removes such spurious features.

The estimated spectral density matrix may not be invertible at all wavenumbers, though this is typically resolved by using more tapers; if it remains non-invertible, a generalized inverse can be used and the results below remain valid. 

\subsection{Partial $K$ function estimation}

Given an estimate of the partial spectra, we estimate the partial $C$ function by first rotationally averaging in wavenumber and then evaluating the corresponding radial inverse transform. 
Rather than directly discretizing the two-dimensional inversion integral in \cref{eq:def:partialc}, we replace the partial spectra by a rotationally averaged approximation and treat this approximation as piecewise constant on radial bins. The resulting radial integral can then be computed exactly. This works well when the Bessel kernel varies more rapidly than the underlying spectra, which is typically the case in practice.

Let $\Omega\subset\RR^2$ denote the finite grid of wavenumbers on which the estimated partial spectra are evaluated. In practice, we rotationally average these estimates onto a one-dimensional radial grid and then evaluate the corresponding radial inverse transform. For some $\waveradiispacing{} > 0$ and $\waveradiimax > 0$, write $\waveradiigrid=(\waveradiispacing/2+\waveradiispacing\ZZ)\cap[0,\waveradiimax]$ for the radial wavenumber grid.
For each $\waveradii\in\waveradiigrid$, define the corresponding shell of wavenumbers by
$
    \shellwavenumbers{\waveradii} = \set{\freq\in\Omega:\norm{\freq}\in(\waveradii-\waveradiispacing/2,\waveradii+\waveradiispacing/2]}.
$
We then set
\begin{align*}
    \mtpgram{XY}^{(rot)}(\waveradii) = \frac{1}{\abs{\shellwavenumbers{\waveradii}}}
    \sum_{\freq\in\shellwavenumbers{\waveradii}} \mtpgram{XY}(\freq),
\end{align*}
the rotationally averaged estimate at $\waveradii\in\waveradiigrid$.
Our partial $C$ function estimate is then
\begin{align*}
    \Cfuncest{XY\pred Z}(r) &= \sum_{\waveradii\in\waveradiigrid}
    \left[\mtpgram{XY\pred Z}^{(\mathrm{rot})}(\waveradii) - \hat\lambda_X\delta_{X,Y}\right] Q_r(\kappa), & r\geq 0,
\end{align*}
where $Q_r(\kappa) = \besselj{0}\left(2\pi r\left[\waveradii-{\waveradiispacing}/{2}\right]\right) - \besselj{0}\left(2\pi r\left[\waveradii+{\waveradiispacing}/{2}\right]\right)$.
For a simple point process, $\reducedmomentmeasure{XY}(\set{0}) = \lambda_X\delta_{X,Y}$, and so $\hat\lambda_X\delta_{X,Y}$ is the natural plug-in estimate of the atom term. The partial $K$ function estimate is then obtained from \cref{eq:c2k} by replacing $\Cfunc{XY}$ with $\Cfuncest{XY\pred Z}$, and similarly the partial pair correlation and partial $L$ functions are obtained from \cref{eq:c2g} and the signed $L$ transformation.

\subsection{Practical implementation}

Additional implementation details, including the choice of hyperparameters, are given in \cref{app:estimation}.
The computational complexity of the estimation procedure is competitive with standard methods for estimating the $K$ function. If $n$ is the total number of points, $P$ the number of processes, $M$ the number of tapers and $R$ the number of spatial distances at which the $K$ function is evaluated, then computing the multitaper periodogram costs $O(PMn\log n)$, up to factors depending on the desired NUFFT tolerance when computing the non-uniform FFT \citep[see, for example,][]{dutt1993fast}. Computing the partial spectra costs $O(P^3|\Omega|)$, where $|\Omega|$ is the number of wavenumbers considered, and computing the $K$ function from the spectral estimate costs $O(R|\Omega|)$. If we use a fixed highest wavenumber and choose the wavenumber grid so that its spacing scales proportionally to the inverse side length of the bounding box of the observational window, then $|\Omega|$ scales as $O(n)$ as the region size grows. Therefore the overall complexity is $O(PMn\log n + P^3 n + Rn)$. Standard direct approaches for computing the $K$ function have complexity $O(P^2 n^2)$, and so for large $n$ and fixed $P$, our approach is faster.

\section{Accounting for covariates and long scale phenomena}\label{sec:covariates}

In many applications, it is also important to be able to account for covariates.
The partial $K$ function developed here can be readily extended to add covariate effects by simply including them as covariate processes. 
In particular, suppose that in addition to the point processes $N_1,\ldots,N_P$, we also observe random fields $Y_1,\ldots,Y_Q$ over the same region. 
For each field $Y_j$, define the corresponding random measure $\xi_j$ by $\xi_j(B) = \int_B Y_j(u) \de u$ for $ B \in \borel{\RR^2}$.
Then the joint system consisting of the point processes $N_1,\ldots,N_P$ and the random measures $\xi_1,\ldots,\xi_Q$ can be treated within the same spectral framework \citep[see][for example]{daley2003introduction,grainger2025spectral}.
In particular, the partial spectra are formed exactly as before from the joint spectral density matrix of this enlarged system, and the resulting partial $C$ and $K$ functions are obtained by the same inverse transformation as in the point process case.
The analogues of \cref{prop:sdf2c,prop:bestprediction,prop:partialcov:isresidualcov} hold for this enlarged system.

Estimation of the required spectra and cross-spectra may be carried out using the methodology of \cite{grainger2025spectral}. The main additional practical issue is that a random field observed on a grid only provides spectral information up to its Nyquist frequency. If relevant atom-corrected point-process spectral content, $\sdf{XY}(\freq)-\reducedmomentmeasure{XY}(\set{0})$, remains beyond this range, one may either interpolate the field onto a finer grid or ignore the unresolved field contribution when forming the predictor. In the latter case, the full spectral matrix is singular, but this is equivalent to treating the unresolved field contribution as a zero-variance (deterministic) component and forming the predictor from the reduced spectral matrix for the remaining components only. A simulated comparison with the inhomogeneous $K$ function is given in \cref{app:simulation:inhom}, where the partial $K$ function again recovers the underlying structure.

\section{Exploratory Analysis of Lansing Woods data}
We illustrate the partial analysis technique on the Lansing Woods data \citep{gerrard1969competition}.
This canonical dataset consists of the locations of trees present in a forest plot in Lansing Woods, Michigan, USA.
There are five named species of tree present, as well as a miscellaneous category.
We account for the miscellaneous category as a covariate process, but do not consider it for analysis as it is hard to interpret biologically (being an aggregate of multiple rare species).
The data has been widely studied in the spatial statistics literature, see \cite{baddeley2015spatstat} for references.
The data is available standardised to the unit square in the \texttt{spatstat} R package \citep{baddeley2015spatstat}.
The top row of \cref{fig:lansing_marginal_analysis} shows the locations of the six different tree species present in the Lansing woods dataset (with one of the species being miscellaneous).
Whilst the individual species seem to be slightly inhomogeneous, our assumptions of homogeneity are about the system as a whole, and looking at one species at a time would be misleading, as it would detect conditional inhomogeneity, and we are not claiming the processes are conditionally homogeneous given the other processes.
However, there could be dependence on unobserved covariates, and thus the results should be interpreted with this caveat in mind.

The bottom row of \cref{fig:lansing_marginal_analysis} shows the estimated marginal $L$ functions and partial $L$ functions for each species with itself, accounting for all other species.
The marginal $L$ functions indicate more clustering behaviour than the partial $L$ functions.
This is not surprising as we expect to see some of the clustering behaviour being explained by interactions with other species \citep{murrell2010does}.
\Cref{fig:lansing_cross_analysis} shows the cross $L$ functions and partial $L$ functions between different species, again accounting for all other species.
Here we again see that some of the clustering behaviour seen in the cross $L$ functions is explained away by the other species.
In some cases, the interaction essentially disappears, even slightly changing sign (e.g. Black Oak vs Hickory).
This general behaviour is expected, as the partial $L$ function has accounted for some of the other processes, meaning that some of the observed clustering/repulsion has been removed by accounting for the covariate processes.
The resulting ``sparser'' representation of the dependence structure between and within the processes can then be used to highlight the more important interactions.

\begin{figure}[!h]
    \centering
    \includegraphics[]{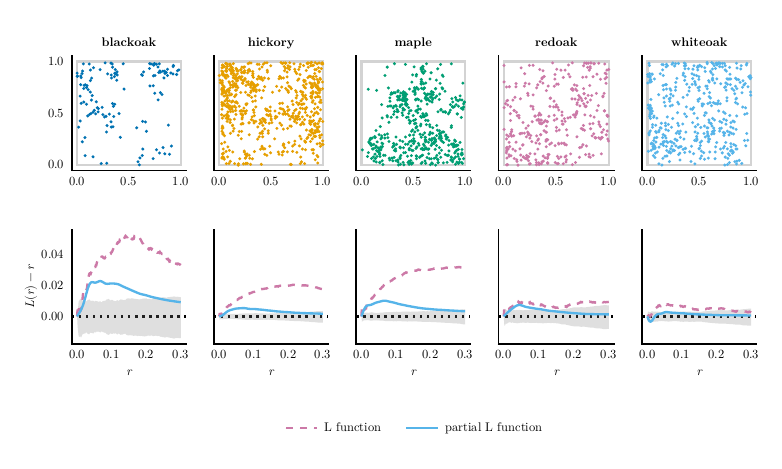}
    \caption{Lansing woods data (top row) and the estimated $L$ and partial $L$ functions within processes (bottom row).}
    \label{fig:lansing_marginal_analysis}
\end{figure}

\begin{figure}[!h]
    \centering
    \includegraphics[]{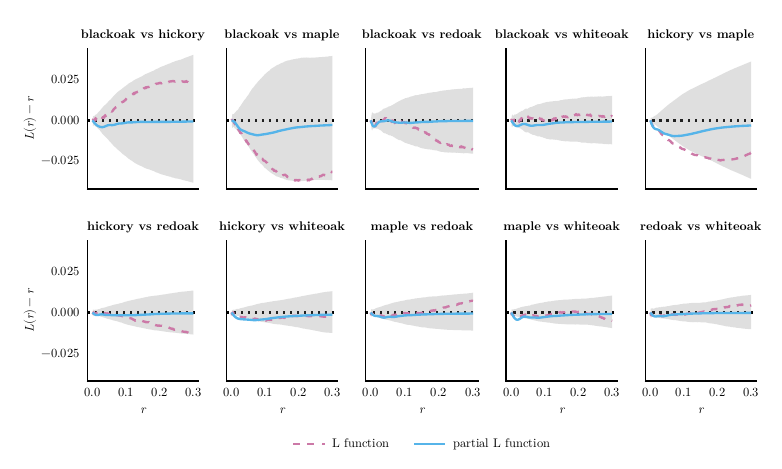}
    \caption{Estimated cross $L$ and partial $L$ functions for the Lansing woods data.}
    \label{fig:lansing_cross_analysis}
\end{figure}

\section{Discussion}

In this paper, we proposed the partial $K$ function, an extension of the usual $K$ function which accounts for some effects of the other processes involved in the system of interest.
The partial $K$ function is non-parametric, computationally efficient, and has a spatial interpretation analogous to that of the usual $K$ function.
It therefore provides a useful additional tool for exploring interactions in multivariate point pattern data.

There are nevertheless several important limitations.
First, the method is based on linear prediction and second-order structure, and so it is not a conditional notion of dependence.
In particular, partial $K$ functions cannot in general be interpreted as detecting conditional orthogonality or conditional independence.
Second, as the number of processes increases, obtaining stable estimates of the spectral matrix and its inverse becomes more challenging \citep{walden2000unified}.
Third, null resampling methods for these partial statistics remain an open problem.

We also assume that the processes are jointly homogeneous.
Care is therefore needed when interpreting the results, since unobserved large-scale covariates may still induce structure that is not accounted for.
At the same time, inhomogeneous methods are not without drawbacks: they require estimation of the intensity function, which can introduce additional variability and sensitivity to the choice of intensity estimator \citep[for example]{baddeley2000non,henrys2009inference}.
The partial $K$ function should therefore be viewed as complementary to inhomogeneous methods rather than as a replacement for them.
Covariate information can also be incorporated into the present framework by treating covariates as random fields, allowing some aspects of inhomogeneity to be addressed within the same methodology \citep{grainger2025spectral}.

There are several natural directions for further work.
One is an extension to marked point processes, where analogous partial summaries could be constructed using the corresponding marked spectral quantities.
Another is the development of anisotropic partial $K$ functions, obtained by replacing the isotropic ball with more general directional sets \citep[for example]{moller2016cylindrical} and using the corresponding Fourier inversion formula.
Both extensions fit naturally within the framework developed here and would broaden the range of applications of the method.

Overall, the partial $K$ function provides a practical way to account for other processes when exploring multivariate point pattern data, while retaining an interpretation on the same spatial scale as the usual $K$ and $L$ functions.

\section{Competing interests}
No competing interest is declared.

\section{Data availability}
The Lansing Woods data is available in the \texttt{R} package \texttt{spatstat} \citep{baddeley2015spatstat}.
The code to reproduce all the simulations and data analysis in this paper is available in the Zenodo repository \cite{jake_p_grainger_2025_17779726}.
Code implementing the methodology in \texttt{Julia} \citep{Julia-2017} is available in the \texttt{SpatialMultitaper.jl} package \citep{grainger_2025_17779550}, with an \texttt{R} \citep{R-2025} frontend package \citep{jake_p_grainger_2026_19594943}.



\section{Acknowledgments}
Sofia Olhede would like to thank the European Research Council under Grant CoG 2015-682172NETS, within the Seventh European Union Framework Program.


\clearpage
\setcounter{section}{0}
\renewcommand\thesection{S\arabic{section}}
\setcounter{equation}{0}%
\setcounter{theorem}{0}
\setcounter{lemma}{0}
\setcounter{corollary}{0}
\setcounter{proposition}{0}
\setcounter{definition}{0}
\setcounter{assumption}{0}
\setcounter{remark}{0}
\setcounter{step}{0}
\setcounter{condition}{0}
\setcounter{property}{0}
\setcounter{restrictions}{0}
\setcounter{example}{0}
\setcounter{algo}{0} 
\renewcommand\thetheorem{\text{S}\arabic{theorem}}
\renewcommand\thelemma{\text{S}\arabic{lemma}}
\renewcommand\thecorollary{\text{S}\arabic{corollary}}
\renewcommand\theproposition{\text{S}\arabic{proposition}}
\renewcommand\thedefinition{\text{S}\arabic{definition}}
\renewcommand\theassumption{\text{S}\arabic{assumption}}
\renewcommand\theremark{\text{S}\arabic{remark}}
\renewcommand\thestep{\text{S}\arabic{step}}
\renewcommand\thecondition{\text{S}\arabic{condition}}
\renewcommand\theproperty{\text{S}\arabic{property}}
\renewcommand\therestrictions{\text{S}\arabic{restrictions}}
\renewcommand\theexample{\text{S}\arabic{example}}
\renewcommand\thealgo{\text{S}\arabic{algo}}

\thispagestyle{empty}
\vbox{%
\hsize\textwidth
\linewidth\hsize
\vskip 0.1in
\makeatletter\@toptitlebar\makeatother
\centering
{\LARGE\sc Supplementary material for \makeatletter\@title\makeatother\par}
\makeatletter\@bottomtitlebar\makeatother
\vskip 0.1in
}
\renewcommand{\shorttitle}{Supplementary material}



\section{Table of notation}
\def\addsymbol #1: #2{$#1$ \> \parbox[t]{0.9\textwidth}{#2}\\}

\begin{tabbing}
    $d$~~~~~~~~~~~~~~~~\=\parbox[t]{0.9\textwidth}{Dimension of space.}\\
    \addsymbol P: {Number of point types.}
    \addsymbol X,Y,Z: {Point types.}
    \addsymbol N: {Ground process containing all point types.}
    \addsymbol N_X: {Point process of type $X$.}
    \addsymbol N_X^0: {Centred count measure, $N_X(A) - \intensity{X}\leb{A}$.}
    \addsymbol \intensity{X}: {Intensity of point process $N_X$.}
    \addsymbol r: {Distance parameter.}
    \addsymbol u: {Spatial lag variable.}
    \addsymbol s: {Spatial location variable.}
    \addsymbol \freq: {Frequency variable.}
    \addsymbol A,B: {Generic measurable sets.}
    \addsymbol \leb{B}: {Lebesgue measure of a Borel set $B$.}
    \addsymbol \RRd: {$d$-dimensional real space.}
    \addsymbol \NN_0: {Non-negative integers.}
    \addsymbol \borel{\RR^d}: {Borel sets on $\RR^d$.}
    \addsymbol \UU^2: {Unit square $[0,1]^2$.}
    \addsymbol \sphere{2}: {Unit ball in $\RR^2$.}
    \addsymbol \zerosphere{2}: {Unit ball minus the origin, $\sphere{2} \setminus \set{0}$.}
    \addsymbol \totallyfinite{\RR^2}: {Set of totally finite signed measures on $\RR^2$.}
    \addsymbol \reducedmomentmeasure{XY}: {Second-order reduced moment measure between point processes $N_X$ and $N_Y$.}
    \addsymbol \reducedcumulantmeasure{XY}: {Second-order reduced cumulant measure between point processes $N_X$ and $N_Y$.}
    \addsymbol \Kfunc{XY}: {Ripley's $K$ function between point processes $N_X$ and $N_Y$.}
    \addsymbol \Lfunc{XY}: {$L$ function between point processes $N_X$ and $N_Y$.}
    \addsymbol \paircf{XY}: {Pair correlation function between point processes $N_X$ and $N_Y$.}
    \addsymbol \Cfunc{XY}: {$C$ function between point processes $N_X$ and $N_Y$.}
    \addsymbol \sdf{XY}: {Spectral density function between point processes $N_X$ and $N_Y$.}
    \addsymbol \EE{\cdot}: {Expectation operator.}
    \addsymbol \var{\cdot}: {Variance operator.}
    \addsymbol \norm{\cdot}: {The $L^2$ norm.}
    \addsymbol \ip{x}{y}: {Dot product between $x$ and $y$.}
    \addsymbol \argmin: {Argument of minimum.}
    \addsymbol \besselj{0}(\cdot): {Bessel function of the first kind of order 0.}
    \addsymbol \besselj{1}(\cdot): {Bessel function of the first kind of order 1.}
    \addsymbol \Kfunc{XY\pred Z}: {Partial $K$ function between point processes $N_X$ and $N_Y$ accounting for point processes $N_Z$.}
    \addsymbol \Predprocess{X\pred Z}^0: {Centred prediction process.}
    \addsymbol \Predprocess{X\pred Z}: {Prediction process for points of type $X$ predicted linearly from point processes $N_Z$.}
    \addsymbol \predprocess{X\pred Z}: {Density of $\Predprocess{X\pred Z}$.}
    \addsymbol \residualprocess{X\pred Z}: {Residual process for points of type $X$ after linear prediction from point processes $N_Z$.}
    \addsymbol \centeredresidualprocess{X\pred Z}: {Centred residual process for points of type $X$ after linear prediction from point processes $N_Z$.}
    \addsymbol \Predkernel{X\pred Z}: {Prediction kernel for points of type $X$ predicted linearly from point processes $N_Z$.}
    \addsymbol \predkernel{X\pred Z}: {Density of $\Predkernel{X\pred Z}$.}
    \addsymbol \predkernelft{X\pred Z}: {Fourier transform of $\Predkernel{X\pred Z}$.}
    \addsymbol \predkernelest{X\pred Z}: {Estimate of $\predkernel{X\pred Z}$.}
    \addsymbol h_m: {The $m$th taper function.}
    \addsymbol \dft{X}{m}: {Tapered Fourier transform of point processes $N_X$ using taper $h_m$.}
    \addsymbol \pgram{XY}{m}: {Single taper periodogram between point processes $N_X$ and $N_Y$ using taper $h_m$.}
    \addsymbol \mtpgram{XY}: {Multitaper periodogram between point processes $N_X$ and $N_Y$.}
    \addsymbol M: {the number of tapers.}
    \addsymbol \wavenumberspacing{j}: {the number of wavenumbers used in the $j$th direction.}
    \addsymbol \wavenumbermax{j}: {the maximum wavenumber used in the $j$th direction.}
    \addsymbol \regionlength{j}: {the length of the bounding box of the observation window in the $j$th direction.}
\end{tabbing}

\section{Theory}
\subsection{Estimators and summary statistics in the $d$-dimensional case}\label{app:generaldim}
In the following subsections, we state the equivalent results in the $d$-dimensional case to those given in the main paper.

\subsubsection{Summary statistics}
Let $\sphere{d}$ be the unit sphere in $\RR^d$ and $\zerosphere{d} = \sphere{d} \setminus \set{0}$.
Write $A_{d-1}$ for the surface area of the unit sphere in $\RR^d$, then the volume and surface area are
\begin{align}
    \leb{\sphere{d}} & = \frac{\pi^{d/2}}{\Gamma(d/2+1)}, \qquad
    A_{d-1} = \frac{2 \pi^{d/2}}{\Gamma(d/2)},
\end{align}
respectively.

Best linear prediction is defined simply by replacing $\RR^2$ and $\UU^2$ with $\RR^d$ and $\UU^d$, respectively, so that
\begin{equation}
    \Predprocess{X\pred Z}^0(A) = \sum_{j=1}^q \int_{\RR^d} \Predkernel{X\pred Z, j}(A-x) N_{Z_j}^0(\de x),
\end{equation}
where
\begin{equation}
    \Predkernel{X\pred Z} \in \argmin_{\Predkernel{} \in \left(\totallyfinite{\RR^d}\right)^{1\times q}} \var{N_X^0(\UU^d) - \sum_{j=1}^q \int_{\RR^d} \Predkernel{j}(\UU^d-x) N_{Z_j}^0(\de x)}.
\end{equation}

The summary statistics are as follows:
\begin{align}
    \intensity{X}                & = \EE{N_X(\UU^d)}                                                                            \\
    \reducedmomentmeasure{XY}(B) & = \EE{\int_{\UU^d} N_X(B+x) N_Y(\de x)},                     \\
    \reducedcumulantmeasure{XY}(B) &= \reducedmomentmeasure{XY}(B) - \intensity{X}\intensity{Y}\leb{B} \\
    \Kfunc{XY}(r)                & = \intensity{X}^{-1} \intensity{Y}^{-1}\reducedmomentmeasure{XY}\left(r\zerosphere{d}\right) \\
    \Lfunc{XY}(r)                & = \mathrm{sgn}(K_{XY}(r))\left(\frac{\abs{K_{XY}(r)}}{\leb{\sphere{d}}}\right)^{1/d}         \\
    \paircf{XY}(r)               & = \frac{\Kfunc{XY}'(r)}{A_{d-1} r^{d-1}}                                                     \\
    \Cfunc{XY}(r)                & = \reducedcumulantmeasure{XY}(r\zerosphere{d})                                               \\
    \sdf{XY}(\freq)              & = \int_{\RR^d} e^{-2\pi i \ip{u}{\freq}} \reducedcumulantmeasure{XY}(\de u)
\end{align}

Some useful relations used in the main paper in the $d$-dimensional case are:
\begin{align}
    \Kfunc{XY}(r)  & = \frac{\Cfunc{XY}(r)}{\intensity{X} \intensity{Y}} + \leb{r\zerosphere{d}}                                                                         \\
    \paircf{XY}(r) & = \frac{\Cfunc{XY}'(r)}{A_{d-1} r^{d-1} \intensity{X} \intensity{Y}} + 1                                                                            \\
    \Cfunc{XY}(r)  & = \int_{\RR^d} \left(\frac{r}{\norm{\freq}}\right)^{d/2} \besselj{d/2}(2\pi\norm{\freq}r) [\sdf{XY}(\freq) - \reducedmomentmeasure{XY}(\set{0})] \de\freq \\
    \Cfunc{XY}'(r) & = 2\pi \int_{\RR^d} \norm{\freq} \left(\frac{r}{\norm{\freq}}\right)^{d/2} \besselj{d/2-1}(2\pi\norm{\freq}r) [\sdf{XY}(\freq) - \reducedmomentmeasure{XY}(\set{0})] \de\freq
\end{align}

\subsubsection{Estimators}\label{app:generaldim:estimator}
For the second estimator we construct, with an additional rotational averaging, we have in the general case
\begin{align}
    \Cfuncest{XY}(r)
        & = \sum_{\kappa \in K} \mtpgram{XY}^{(rot)}(\kappa) \int_{\RRd} \indicator_{(\kappa-\frac{s}{2},\kappa+\frac{s}{2}]}(\norm{\freq}) \left(\frac{r}{\norm{\freq}}\right)^{d/2} \besselj{d/2}(2\pi \norm{\freq}r) \de \freq \\
        & = \sum_{\kappa \in K} \mtpgram{XY}^{(rot)}(\kappa) A_{d-1} r^{d/2} \int_{\kappa-s/2}^{\kappa+s/2} {x^{d/2-1}} \besselj{d/2}(2\pi xr) \de x \\
        &= \sum_{\kappa \in K} \mtpgram{XY}^{(rot)}(\kappa) [w_d(r,\kappa+s/2) - w_d(r,\kappa-s/2)]
\end{align}
where
\begin{align}
    w_d(r,y) &= A_{d-1} r^{d/2} \int_0^y x^{d/2-1} \besselj{d/2}(2\pi xr) \de x \\
    & = \frac{2 (\pi r y)^d}{\Gamma(d/2+1)^2} \;\oneFtwo{d/2}{d/2+1,d/2+1}{-(\pi r y)^2}.
\end{align}
When $d=3$ because $\besselj{3/2}(x) = \sqrt{2/\pi x}(\sin(x)/x-\cos(x))$ (see \cite{abramowitz1948handbook} 10.1.11, for example),
\begin{align}
    w_3(r,x) &= 
    4 \pi r^{3/2} \int_0^y x^{1/2} \besselj{3/2}(2\pi r x) \de x \\
    &= 4 \pi r^{3/2} \int_0^y x^{1/2} \sqrt{\frac{2}{2\pi^2 r x}}\left(\frac{\sin(2\pi r x)}{2\pi r x} - \cos(2\pi r x)\right) \de x \\
    &= 4 r \int_0^y \frac{\sin(2\pi r x)}{2\pi r x} - \cos(2\pi r x) \de x \\
    &= \frac{2}{\pi} \int_0^{2\pi r y} \frac{\sin(z)}{z} - \cos(z) \de z \\
    &= \frac{2}{\pi} (\mathrm{Si}(2\pi r x) - \sin(2\pi r x))
\end{align}
where 
\begin{align}
    \mathrm{Si}(x) = \int_{0}^x \frac{\sin(y)}{y} \de y.
\end{align}
For completeness, when $d=1$,
\begin{align}
    w_1(r,x) = 2\mathrm{Si}(2\pi rx)/\pi.
\end{align}

\subsection{Proof of \cref{prop:sdf2c}}\label{app:spectratocfunction}
\begin{proof}
Following \cite{daley2003introduction}, write
\begin{equation}
    \tilde{e}_\lambda(k) = \prod_{j=1}^d \frac{\lambda^2}{\lambda^2+k_j^2}.
\end{equation}

Then we have the inverse relation
\begin{equation}
    \reducedcumulantmeasure{XY}(B) = \lim_{\lambda \rightarrow \infty} \int_{\RR^d} \tilde{e}_\lambda(\freq) \tilde\indicator_{B}(\freq) \sdf{XY}(\freq) \de\freq,
\end{equation}
for continuity sets of $\reducedcumulantmeasure{XY}$ \citep{daley2003introduction}.
Define $D_{XY}(B) = \reducedcumulantmeasure{XY}(B\cap \set{0})$, then
\begin{equation}
    D_{XY}(B) = \lim_{\lambda \rightarrow \infty} \int_{\RR^d} \tilde{e}_\lambda(\freq) \tilde\indicator_{B}(\freq) \reducedmomentmeasure{XY}(\set{0}) \de\freq,
\end{equation}
for continuity sets of $D_{XY}$ (which includes continuity sets of  $\reducedcumulantmeasure{XY}$).
Therefore, we have
\begin{align}
    \reducedcumulantmeasure{XY}(B\setminus \set{0}) &= \reducedcumulantmeasure{XY}(B) - D_{XY}(B) \\
    &= \lim_{\lambda \rightarrow \infty} \int_{\RR^d} \tilde{e}_\lambda(\freq) \tilde\indicator_{B}(\freq) \left[\sdf{XY}(\freq) - \reducedmomentmeasure{XY}(\set{0})\right] \de\freq \\
    &= \int_{\RR^d} \tilde\indicator_{B}(\freq) \left[\sdf{XY}(\freq) - \reducedmomentmeasure{XY}(\set{0})\right] \de\freq,
\end{align}
provided that $\sdf{XY}(\freq) - \reducedmomentmeasure{XY}(\set{0})$ is integrable.

As a result, we can recover the reduced cumulant measure from the spectral density function, meaning that we could also recover Ripley's $K$ function from the spectral density function and the intensities.
We have
\begin{align}
    \tilde\indicator_{r\sphere{d}}(\freq) & = 
    \begin{cases}
        \leb{r\sphere{d}} & \text{if } \freq=0, \\
        (r/\norm{\freq})^{d/2} \besselj{d/2}(2\pi\norm{\freq}r) & \text{otherwise.}
    \end{cases}
\end{align}
Furthermore
\begin{align}
    \frac{\partial}{\partial r} \tilde\indicator_{r\sphere{d}}(\freq) & = 
    \begin{cases}
        A_{d-1} r^{d-1} & \text{if } \freq=0,\\                                                      {2\pi}\norm{\freq} (r/\norm{\freq})^{d/2} \besselj{d/2-1}(2\pi\norm{\freq}r) & \text{otherwise.}
    \end{cases}
\end{align}

Therefore because $[\sdf{XY}(\freq) - \reducedmomentmeasure{XY}(\set{0})]$ is integrable, we can interchange limits by Leibniz integral rule and
\begin{align}
    \Cfunc{XY}(r)  & = \int_{\RR^d} \left(\frac{r}{\norm{\freq}}\right)^{d/2} \besselj{d/2}(2\pi\norm{\freq}r) [\sdf{XY}(\freq) - \reducedmomentmeasure{XY}(\set{0})] \de\freq, \\
    \Cfunc{XY}'(r) & = 2\pi r \int_{\RR^d} \left(\frac{r}{\norm{\freq}}\right)^{d/2-1} \besselj{d/2-1}(2\pi\norm{\freq}r) [\sdf{XY}(\freq) - \reducedmomentmeasure{XY}(\set{0})](\freq) \de\freq.
\end{align}
In particular, when $d=2$
\begin{align}
    \Cfunc{XY}(r)  & = \int_{\RR^2} \frac{r}{\norm{\freq}} \besselj{1}(2\pi\norm{\freq}r) [\sdf{XY}(\freq) - \reducedmomentmeasure{XY}(\set{0})] \de\freq, \\
    \Cfunc{XY}'(r) & = 2\pi r \int_{\RR^2} \besselj{0}(2\pi\norm{\freq}r) [\sdf{XY}(\freq) - \reducedmomentmeasure{XY}(\set{0})] \de\freq.
\end{align}
\end{proof}

\subsection{Proof of \cref{prop:bestprediction}}\label{app:bestprediction}
\begin{proof}
Let
\begin{align}
\Predkernel{}=(\Predkernel{1},\ldots,\Predkernel{q})\in (S_{\RR^d})^{1\times q},
\qquad
\predkernelft{}(\freq)=(\predkernelft{1}(\freq),\ldots,\predkernelft{q}(\freq)),
\end{align}
where
\begin{align}
\predkernelft{j}(\freq)=\int_{\RR^d} e^{-2\pi i \ip{x}{\freq}}\,\Predkernel{j}(\de x),
\qquad j=1,\ldots,q.
\end{align}
Consider the prediction error
\begin{align}
N_X^0(\UU^d)-\sum_{j=1}^q \int_{\RR^d} \Predkernel{j}(\UU^d-x)\,N_{Z_j}^0(\de x).
\end{align}
Since both $N_X^0(\UU^d)$ and $\int_{\RR^d} \Predkernel{j}(\UU^d-x)\,N_{Z_j}^0(\de x)$ are centred, the mean squared error of the prediction is equal to the variance of the prediction error.
It is easier to work with the variance of the prediction error in the Fourier domain, so we will calculate the variance of the prediction error by calculating the variance and covariance terms in the Fourier domain.
In particular, we have
\begin{align}
&\var{N_X^0(\UU^d)-\sum_{j=1}^q \int_{\RR^d} \Predkernel{j}(\UU^d-x)\,N_{Z_j}^0(\de x)} \nonumber\\
&\quad = \var{N_X^0(\UU^d)} +\var{\sum_{j=1}^q \int_{\RR^d} \Predkernel{j}(\UU^d-x)\,N_{Z_j}^0(\de x)} \nonumber\\
&\qquad -2\Re\cov{N_X^0(\UU^d),\sum_{j=1}^q \int_{\RR^d} \Predkernel{j}(\UU^d-x)\,N_{Z_j}^0(\de x)}.
\end{align}
Proceeding term by term, first
\begin{align}
\var{N_X^0(\UU^d)}
&= \int_{\RR^{2d}} \indicator_{\UUd}(x)\indicator_{\UUd}(y)\,\cumulantmeasure{XX}(\de x\times \de y) \\
&= \int_{\RR^d}\abs{\tilde\indicator_{\UUd}(\freq)}^2 \sdf{XX}(\freq)\,\de\freq.
\end{align}
Second, expanding the variance as a double sum over the components of $Z$, and applying the same covariance calculation componentwise, we obtain
\begin{align}
&\var{\sum_{j=1}^q \int_{\RR^d} \Predkernel{j}(\UU^d-x)\,N_{Z_j}^0(\de x)} \\
&\quad= 
\sum_{j=1}^q \sum_{l=1}^q \cov{\int_{\RR^d} \Predkernel{j}(\UU^d-x)\,N_{Z_j}^0(\de x),\int_{\RR^d} \Predkernel{l}(\UU^d-y)\,N_{Z_l}^0(\de y)}\\
&\quad=
\sum_{j=1}^q \sum_{l=1}^q \int_{\RR^d}\abs{\tilde\indicator_{\UUd}(\freq)}^2
\predkernelft{j}(\freq)\sdf{Z_jZ_l}(\freq) \conj{\predkernelft{l}(\freq)} \de\freq \\
&\quad=
\int_{\RR^d}\abs{\tilde\indicator_{\UUd}(\freq)}^2
\predkernelft{}(\freq)\sdf{ZZ}(\freq)\predkernelft{}(\freq)^H \de\freq.
\end{align}
Finally,
\begin{align}
&\cov{N_X^0(\UU^d),\sum_{j=1}^q \int_{\RR^d} \Predkernel{j}(\UU^d-x)\,N_{Z_j}^0(\de x)} \\
&\quad=
\int_{\RR^d}\abs{\tilde\indicator_{\UUd}(\freq)}^2
\sdf{XZ}(\freq)\predkernelft{}(\freq)^H\,\de\freq.
\end{align}
Therefore
\begin{align}
&\var{N_X^0(\UU^d)-\sum_{j=1}^q \int_{\RR^d} \Predkernel{j}(\UU^d-x)\,N_{Z_j}^0(\de x)} \nonumber\\
&\quad=
\int_{\RR^d}\abs{\tilde\indicator_{\UUd}(\freq)}^2
\Big(
\sdf{XX}(\freq)
-2\Re\{ \sdf{XZ}(\freq)\predkernelft{}(\freq)^H \}
+\predkernelft{}(\freq)\sdf{ZZ}(\freq)\predkernelft{}(\freq)^H
\Big)\,\de\freq.
\label{eq:varwavenumber:multivariate}
\end{align}
Since $\abs{\tilde\indicator_{\UUd}(\freq)}^2\ge 0$ and is non-zero for almost every
$\freq\in\RR^d$, the integral is minimized by minimizing the bracketed term in
\cref{eq:varwavenumber:multivariate} for almost every $\freq$. For each such $\freq$,
the minimizer is
\begin{align}
\predkernelft{}(\freq)=\sdf{XZ}(\freq)\sdf{ZZ}(\freq)^{-1}.
\end{align}
By \cref{assumption:predictionkernel}, this Fourier transform corresponds to
$\Predkernel{X\pred Z}$, and hence $\Predkernel{X\pred Z}$ belongs to the argmin.
\end{proof}

\subsection{Proof of \cref{prop:partialcov:isresidualcov}}\label{sec:partialcov:isresidualcov}

\begin{proof}
    We define the centred residual process by $\centeredresidualprocess{X\pred Z}(B)=\residualprocess{X\pred Z}(B)-\intensity{X}\leb{B}$. For all $r\geq 0$,
    \begin{align}
        \Cfunc{\residualprocess{X\pred Z}, Y}(r) 
        &= \EE{\int_{\UU^2} \centeredresidualprocess{X\pred Z}(r\zerosphere{2}+y) N_Y(\de y)} \\
        &= \Cfunc{\centeredresidualprocess{X\pred Z}, Y}(r).
    \end{align}
    Therefore, 
    \begin{align}
        \Cfunc{\residualprocess{X\pred Z}, Y}(r)
        &= \Cfunc{\centeredresidualprocess{X\pred Z}, Y}(r) \\
        &= \EE{\int_{\UU^2} \centeredresidualprocess{X\pred Z}(r\zerosphere{2}+y) N_Y(\de y)} \\
        &= \EE{\int_{\UU^2} N_X^0(r\zerosphere{2}+y) N_Y(\de y)}
        - \EE{\int_{\UU^2} \Predprocess{X\pred Z}^0(r\zerosphere{2}+y) N_Y(\de y)} \\
        &= \Cfunc{XY}(r) - \sum_{j=1}^q \EE{\int_{\UU^2}\int_{\RR^2} \Predkernel{X\pred Z,j}(r\zerosphere{2}-y+z) N_{Z_j}^0(\de z) N_Y(\de y)} \\
        &= \Cfunc{XY}(r) - \sum_{j=1}^q S_j(r)
    \end{align}
    Now for all $j=1,\ldots,q$,
    \begin{align}
        S_j(r)
        &= \EE{\int_{\UU^2}\int_{\RR^2} \Predkernel{X\pred Z,j}(r\zerosphere{2}-y+z) N_{Z_j}^0(\de z) N_Y(\de y)} \\
        &=  \EE{\int_{\UU^2}\int_{\RR^2}\int_{\RR^2} \indicator_{r\zerosphere{2}}(x-y+z) \Predkernel{X\pred Z, j}(\de x) N_{Z_j}^0(\de z) N_Y(\de y)} \\
        &= \int_{\RR^2}  \EE{\int_{\UU^2} N_{Z_j}^0(r\zerosphere{2}+y-x) N_Y(\de y)} \Predkernel{X\pred Z, j}(\de x) \label{eq:residual:exchange:1} \\
        &= \int_{\RR^2} \reducedcumulantmeasure{Z_jY}(r\zerosphere{2}-x) \Predkernel{X\pred Z, j}(\de x) \\
        &= \int_{\RR^2} \int_{\RR^2} \tilde\indicator_{r\zerosphere{2}-x}(-\freq)\left[f_{Z_jY}(\freq)-\reducedmomentmeasure{Z_jY}(\set{0})\right] \de\freq \Predkernel{X\pred Z, j}(\de x) \\
        &= \int_{\RR^2} \int_{\RR^2} \tilde\indicator_{r\zerosphere{2}}(\freq) e^{-2\pi i \ip{\freq}{x}} f_{Z_jY}(\freq) \de\freq \Predkernel{X\pred Z, j}(\de x) \\
        &= \int_{\RR^2} \tilde\indicator_{r\zerosphere{2}}(\freq) \predkernelft{X\pred Z, j}(\freq) f_{Z_jY}(\freq) \de\freq. \label{eq:residual:exchange:2}
    \end{align}
    Therefore
    \begin{align}
        \Cfunc{\centeredresidualprocess{X\pred Z}, Y}(r) 
        &= \Cfunc{XY}(r) - \sum_{j=1}^q \int_{\RR^2} \tilde\indicator_{r\zerosphere{2}}(\freq) \predkernelft{X\pred Z, j}(\freq) f_{Z_jY}(\freq) \de\freq \\
        &= \Cfunc{XY}(r) - \int_{\RR^2} \tilde\indicator_{r\zerosphere{2}}(\freq) \predkernelft{X\pred Z}(\freq) f_{ZY}(\freq) \de\freq \\
        &= \int_{\RR^2} [\sdf{XY}(\freq)-\reducedmomentmeasure{XY}(\set{0})] \frac{r}{\norm{\freq}}\besselj{1}(2\pi\norm{\freq}r) \de\freq \\
        &\qquad - \int_{\RR^2} \sdf{XZ}(\freq) \sdf{ZZ}(\freq)^{-1} \sdf{ZY}(\freq) \frac{r}{\norm{\freq}}\besselj{1}(2\pi\norm{\freq}r) \de\freq \\
        &= \int_{\RR^2} [\sdf{XY\pred Z}(\freq)-\reducedmomentmeasure{XY}(\set{0})] \frac{r}{\norm{\freq}}\besselj{1}(2\pi\norm{\freq}r) \de\freq \\
        &= \Cfunc{XY\pred Z}(r).
    \end{align}
    The interchange of limits in \cref{eq:residual:exchange:1} is justified provided
    \begin{align}
        T = \int_{\RR^2} \EE{\int_{\UU^2} \abs{N_{Z_j}^0}(r\zerosphere{2}+y-x) N_Y(\de y)} \abs{\Predkernel{X\pred Z, j}}(\de x) <\infty.
    \end{align}
    This holds because
    \begin{align}
        T &\leq \int_{\RR^2} \EE{\int_{\UU^2} N_{Z_j}(r\zerosphere{2}+y-x) N_Y(\de y)} \abs{\Predkernel{X\pred Z, j}}(\de x) \\ &\quad+ \EE{\int_{\UU^2} \lambda_{Z_j} \leb{r\zerosphere{2}+y-x} N_Y(\de y)} \abs{\Predkernel{X\pred Z, j}}(\de x) \\
        &= T_1 + T_2,
    \end{align}
    and we have
    \begin{align}
        T_1 & = \int_{\RR^2} \EE{\int_{\UU^2} N_{Z_j}(r\zerosphere{2}+y-x) N_Y(\de y)} \abs{\Predkernel{X\pred Z, j}}(\de x) \\
        & = \int_{\RR^2} \reducedmomentmeasure{Z_jY}(r\zerosphere{2}-x) \abs{\Predkernel{X\pred Z, j}}(\de x) \\
        & \leq C_{r\zerosphere{2}} \int_{\RR^2} \abs{\Predkernel{X\pred Z, j}}(\de x) \\
        & = C_{r\zerosphere{2}} \abs{\Predkernel{X\pred Z, j}}(\RR^2) \\
        & <\infty,
    \end{align}
    where $C_{r\zerosphere{2}}<\infty$ is a constant and comes from the translation boundedness of the reduced moment measure \citep[Proposition 8.3.I (iv)]{daley2003introduction}, and $\abs{\Predkernel{X\pred Z, j}}(\RR^2)<\infty$ by the assumption of total finiteness.
    The second term is also finite because
    \begin{align}
        T_2 & = \int_{\RR^2} \EE{\int_{\UU^2} \lambda_{Z_j} \leb{r\zerosphere{2}+y-x} N_Y(\de y)} \abs{\Predkernel{X\pred Z, j}}(\de x) \\
        & = \int_{\RR^2} \lambda_{Z_j} \leb{r\zerosphere{2}} \EE{N_Y(\UU^2)} \abs{\Predkernel{X\pred Z, j}}(\de x) \\
        &= \abs{\Predkernel{X\pred Z, j}}(\RR^2) \leb{r\zerosphere{2}} \lambda_{Z_j} \lambda_Y \\
        &<\infty.
    \end{align}

    The interchange of limits in \cref{eq:residual:exchange:2} is justified because
    \begin{align}
        &\hspace{-10em} \int_{\RR^2} \int_{\RR^2} \abs{\tilde\indicator_{r\zerosphere{2}}(\freq) e^{-2\pi i \ip{\freq}{x}} f_{ZY}(\freq)} \de\freq \abs{\Predkernel{X\pred Z, j}}(\de x) \nonumber \\
        &\leq \int_{\RR^2} \int_{\RR^2} \abs{f_{ZY}(\freq)} \de\freq \abs{\Predkernel{X\pred Z, j}}(\de x) \\
        &= \abs{\Predkernel{X\pred Z, j}}(\RR^2) \int_{\RR^2} \abs{f_{ZY}(\freq)} \de\freq \\
        &<\infty,
    \end{align}
    as $f_{ZY}(\freq) = f_{ZY}(\freq) - \reducedmomentmeasure{ZY}(\set{0})$ is integrable by assumption, and $\abs{\Predkernel{X\pred Z, j}}(\RR^2)<\infty$ by the assumption of total finiteness.
\end{proof}

\section{Estimation}\label{app:estimation}
\subsection{Spectral estimation}\label{app:estimation:spectral}
We assume that we observe data on some bounded observational region $\region$.
To construct estimators for the partial $K$ function, we first need estimators of the spectral density function, and then the partial spectral density function.
We will use the multitaper method \citep{thomson1982spectrum, grainger2025spectral} as this allows us to construct reliable estimators of the spectral density function from observations on arbitrary domains.
Briefly, given a family of tapers, $h_1,\ldots,h_M$ supported on a subset of $\region$, we construct tapered periodograms
\begin{align}
    \pgram{XY}{m}(\freq) & = \dft{X}{m}(\freq) \conj{\dft{Y}{m}(\freq)}, & \freq \in \RR^2,
\end{align}
where
\begin{align}
    \dft{X}{m}(\freq) & = \int_{\region} h_m(x) e^{-2\pi i \ip{x}{\freq}} N_X(\de x) - \hat\lambda_X H_m(\freq), & \freq \in \RR^2,
\end{align}
with $H_m(\freq)$ being the Fourier transform of the taper $h_m$ at wavenumber $\freq$.
Tapers are constructed either by interpolated Slepian tapers (on arbitrary domains) or by taking outer products of the minimum bias tapers (on rectangular domains), see \cite{grainger2025spectral} for details.
The multitaper periodogram is subsequently defined as
\begin{align}
    \mtpgram{XY}(\freq) & = \frac{1}{M} \sum_{m=1}^M \pgram{XY}{m}(\freq), & \freq \in \RR^2.
\end{align}
This provides a reliable estimator of the spectral density function, whose properties are well understood \citep[see e.g.,][]{thomson1982spectrum, walden2000unified, grainger2025spectral}.

\subsection{Basic estimator}\label{app:estimation:basic}
In order to estimate the (partial) $K$ function, we first estimate the (partial) $C$ function and then use a discretization of \cref{eq:sdf2c}.
In other words, we estimate the (partial) $C$ function with a truncated discretisation of the inversion integral, applied to the (partial) spectral density function estimator.
This technique can be used to estimate both the usual $K$ function and the partial $K$ function.
We therefore describe it in terms of the usual case, but the (simple) partial case is a plug in of the partial spectral density function estimate (see \cref{app:estimation:debias} for additional debiasing).

In order to use \cref{eq:sdf2c}, we need to discretize and truncate the integral.
We do this by considering a finite set of wavenumbers $\Omega\subset\RR^2$ and then approximating the integral by a Riemann sum.
In particular, let $\wavenumbermax{}\in(0,\infty)^2$ be the maximum wavenumber and $\wavenumberspacing{}\in(0,\infty)^2$ be the spacing between wavenumbers in each dimension.
Then $\Omega = (\wavenumbermax{} \circ [-1,1]^2) \cap (\wavenumberspacing{}\circ \ZZ^2)$.\footnote{In practice, the output wavenumber grid corresponds to an FFT, and thus takes a specific form, but this is not important for the exposition here.}
We construct the estimator for $C$ by
\begin{align}
    \Cfuncest{XY}(r) & = \wavenumberspacing{1} \wavenumberspacing{2} \sum_{\freq \in \Omega} \left[\mtpgram{XY}(\freq) -\intensityest{X}\delta_{X,Y}\right] \frac{r}{\norm{\freq}}\besselj{1}(2\pi\norm{\freq}r), & r\geq 0.
\end{align}
In general, the term $\intensityest{X}\delta_{X,Y}$ should be an estimator of $\reducedmomentmeasure{XY}(\set{0})$, which is $\intensity{X}\delta_{X,Y}$ for a simple point process, but takes a different form for a marked point process for example, i.e. it is not the first moment of a random measure in general \citep{daley2003introduction}.

\subsection{Debiasing}\label{app:estimation:debias}
To obtain an estimate of the partial $K$ function, we replace the estimated spectral density function with its partial counterpart.
The naive estimator is a plug in, so that
\begin{align}
    \mtpgram{XY\pred Z}(\freq) & = \mtpgram{XY}(\freq) - \mtpgram{XZ}(\freq) \mtpgram{ZZ}(\freq)^{-1} \mtpgram{ZY}(\freq), & \freq \in \RR^2.\label{eq:naive_partial_estimator}
\end{align}
The estimated spectral density matrix function may not be invertible at all wavenumbers. In practice, this issue is often resolved by using more tapers, which increases smoothing in wavenumber space and typically ensures invertibility. If the matrix remains non-invertible, a generalized inverse, such as the Moore-Penrose pseudoinverse, can be used. The results stated below remain valid when a generalized inverse is applied, but for simplicity, we present them as if the usual matrix inverse is used throughout.

It is important to note, however, that even if we have unbiased estimates of the spectral density function, \cref{eq:naive_partial_estimator} is still a biased estimator of the partial spectra. This bias can, in turn, lead to substantial bias in the estimated $K$ function.
In particular, consider the function which maps spectral density matrix functions to their partial equivalent
\begin{align}
    \sdftopartial{XY\pred Z}[f] = \sdf{XY} - \sdf{XZ} \sdf{ZZ}^{-1} \sdf{ZY}
\end{align}
where $\sdf{ZZ}^{-1}$ refers to pointwise matrix inversion (not the inverse of the function $\sdf{ZZ}$).
Then given some regularity conditions, for fixed $M$, $P$, and growing domain, the plug in estimator satisfies
\begin{align}
    \EE{\sdftopartial{XY\pred Z}[\hat{f}]} = \left(1-\frac{P_Z}{M}\right) \sdftopartial{XY\pred Z}[f] + o(1)
\end{align}
where $P_Z$ is the number of processes in $Z$ and $M$ is the number of tapers used to construct $\hat{f}$  (see \cref{theorem:bias:partial}).
Therefore, we can obtain an improved estimate by setting
\begin{align}
    \mtpgram{XY\pred Z} & = \left(\frac{M}{M-P_Z}\right) \sdftopartial{XY\pred Z}[\hat{f}].
\end{align}
This is similar to the bias corrections required for estimating partial coherence \citep{medkour2009graphical}, though here we are directly interested in the partial spectra, not the partial coherence as is typically the case in other applications.
Recall that we already require more tapers than processes, and therefore $M>P_Z$.
When we compute partial $L$ functions, we perform another non-linear transformation, which can warp this bias in unusual ways especially for short distances, which results in bias that is not just a percentage reduction, but which looks like a meaningful feature.
A specific example of this phenomenon is given in \cref{app:simulation:debiasing}.
The debiasing we propose resolves this problem, and removes such spurious features.

\subsection{Pre-rotational averaging}\label{app:estimation:preaverage}

So far, we have constructed our estimator by making direct use of \cref{eq:sdf2c} and plugging in our (debiased) estimate of the partial spectra.
However, an alternative approach would be to first construct a rotationally averaged spectral estimate, and then compute a one-dimensional transform.
The rationale for this approach is, first, that we can reduce the number of evaluations of the Bessel functions, which are expensive.
Second, because Bessel functions oscillate, the quality of the integral approximation can be poor if we cannot discretize enough.
This can be tricky, because for small radii, we need to use many wavenumbers before the Bessel function gets small, but for large radii, we need a fine grid to capture their oscillatory behaviour.
This results not so much in a new estimator, but rather a more numerically stable way to discretise the integral in \cref{eq:sdf2c}.

Consider then a simple rotationally averaged estimate $\mtpgram{XY}^{(rot)}$ of the rotationally averaged spectral density function, such as that in \cref{eq:app:rotational_average_estimator}, which is a function of the radial wavenumber $\waveradii = \norm{\freq}$ only.
Notice that we are not assuming the spectral density function is isotropic, but rather that we are just performing a rotational average to obtain a one-dimensional function of the radial wavenumber, which would estimate the rotationally averaged spectral density function.
When applied to the partial spectra, we simply rotationally average $\mtpgram{XY\pred Z}$.
In other words, we do not perform the rotational averaging on the original spectral density function and then compute the partial spectra, but rather we compute the partial spectra first and then perform the rotational averaging.
Performing the rotational averaging first and then computing the partial spectra would only be valid under the additional assumption that the spectral density function is isotropic, which is not necessarily the case, and thus we do not make this assumption here.

Finally then, a useful, faster and more stable alternative may be derived by considering
\begin{align}
    \Cfuncest{XY}(r)
    & = \sum_{\waveradii \in \waveradiigrid} \left[\mtpgram{XY}^{(rot)}(\waveradii)-\intensityest{X}\delta_{X,Y}\right] \int_{\RR^2} \indicator_{(\waveradii-\frac{\waveradiispacing}{2},\waveradii+\frac{\waveradiispacing}{2}]}(\norm{\freq}) \frac{r}{\norm{\freq}} \besselj{1}(2\pi \norm{\freq}r)\de \freq \nonumber \\
    & = \sum_{\waveradii \in \waveradiigrid} \left[\mtpgram{XY}^{(rot)}(\waveradii)-\intensityest{X}\delta_{X,Y}\right] 2\pi r \int_{\waveradii-\waveradiispacing/2}^{\waveradii+\waveradiispacing/2} \besselj{1}(2\pi xr)\de x \nonumber \\
    & = \sum_{\waveradii \in \waveradiigrid} \left[\mtpgram{XY}^{(rot)}(\waveradii)-\intensityest{X}\delta_{X,Y}\right] \left[\besselj{0}(2\pi r (\waveradii-\waveradiispacing/2)) - \besselj{0}(2\pi r (\waveradii+\waveradiispacing/2))\right]\label{eq:improved_estimator}
\end{align}
where $\waveradiigrid = (\waveradiispacing/2+\waveradiispacing \ZZ) \cap [0, \waveradiimax]$ is a one dimensional vector of radial wavenumbers spaced $\waveradiispacing>0$ apart so that the first radial wavenumber considered is bigger than $\waveradiispacing/2$.
The equivalent form for the $d$-dimensional case is given in \cref{app:generaldim:estimator}.
Notice that in \cref{eq:improved_estimator}, we are approximating the integral with the analytical integral assuming that $\mtpgram{XY}^{(rot)}$ is piecewise constant, and not with a Riemann sum.
This results in greater stability, as the underlying spectra is often reasonably smooth relative to the weighted Bessel function it is integrated against.

\subsection{Hyperparameter selection}
In order to estimate the partial $K$ function, we need to select appropriate hyperparameters, namely, the number of tapers $M$, the highest wavenumber $\wavenumbermax{}$, and the spacing $\wavenumberspacing{}$.
The number of tapers, $M$, needs to be at least the number of processes, $P$, i.e. $M \geq P$, because otherwise we will not be able to invert the spectral matrix \citep{walden2000unified}.
Generally, since we aggregate to compute the partial $K$ function, we do not need the tapers for variance reduction as much as we would if we were interested in estimating the partial spectra.
Increasing the number of tapers corresponds to smoothing over a larger bandwidth in wavenumber, which will introduce bias.
One simple option for selecting $\wavenumberspacing{}$ is to use $1 \oslash \regionlength{}$ where $\regionlength{}$ is the vector of side lengths of the bounding box of $\region$ and $\oslash$ denotes elementwise division.
To select $\wavenumbermax{}$, we can either look at pilot estimates of the spectral density function, or use an iterative scheme where we repeatedly increase $\wavenumbermax{}$ until we see a small change in the resultant $K$ function.
Another useful diagnostic check is to compare the $K$ function estimated using standard methods to the $K$ function (not partial) estimated from the spectral density function, which we would expect to be similar.
In \cref{app:simulation:l_function}, we show in simulations that this spectral approach to estimating the $L$ function is competitive with the standard border correction methods typically used when computing the usual $L$ function.

The additional parameters in the pre rotational averaging approach are the radial wavenumber spacing $\waveradiispacing$ and maximum radial wavenumber $\waveradiimax$.
One choice is to relate them to the previous parameters by $\waveradiimax = \min_j{\wavenumbermax{j}}$ and $\waveradiispacing = \min_j{\wavenumberspacing{j}}$.

The final choice is the form of rotational averaging to use. Since throughout this paper we use shell averaging, we define the rotational average directly by averaging over radial shells rather than introducing a more general kernel smoother. For each $\waveradii\in\waveradiigrid$, define the corresponding shell of wavenumbers by
\begin{align}
    \shellwavenumbers{\waveradii}
    =
    \set{\freq\in\Omega:\norm{\freq}\in(\waveradii-\waveradiispacing/2,\waveradii+\waveradiispacing/2]}.
\end{align}
We then set
\begin{align}
    \mtpgram{XY}^{(rot)}(\waveradii)
    =
    \frac{1}{\abs{\shellwavenumbers{\waveradii}}}
    \sum_{\freq\in\shellwavenumbers{\waveradii}} \sdf{XY}(\freq).
    \label{eq:app:rotational_average_estimator}
\end{align}

This is exactly the rotational averaging used in the simulations and applications. It corresponds to replacing the spectra by a piecewise constant function on radial shells of width $\waveradiispacing$, which is then integrated exactly in \cref{eq:improved_estimator}.

\subsection{Computational complexity}\label{app:estimation:complexity}
The computational complexity of the estimation procedure is competitive with standard methods for estimating the $K$ function.
In particular, say that $n$ is the total number of points, $P$ is the number of processes, $M$ is the number of tapers and $R$ is the number of spatial distances at which we want to evaluate the $K$ function.
Then the complexity of computing the multitaper periodogram is $O(PMn\log n)$ (up to factors depending on the desired NUFFT tolerance when computing the non-uniform FFT, see \cite{dutt1993fast} for example).
The complexity of computing the partial spectral density function is $O(P^3 \abs{\Omega})$ where $\abs{\Omega}$ is the number of wavenumbers considered, and the complexity of computing the $K$ function from the spectral density function is $O(R\abs{\Omega})$.
If we use a fixed highest wavenumber and the spacing rule proposed above, then $\abs{\Omega}$ scales as $O(n)$ (as $n$ scales like the region size).
Therefore, the overall complexity is $O(PMn\log n + P^3 n + R n)$.
The standard approaches for computing the $K$ function have complexity $O(P^2 n^2)$, and so for large $n$ and fixed $P$, our approach is faster.

\section{Additional methodological details}

\subsection{Bias correction}

Say that we have $P$ point processes.
In addition, we use a multitaper spectral estimate with a fixed $M > P$ tapers.
We will be interested in a growing domain framework, as in \cite{grainger2025spectral}.
In order to do this, we will write $\mtpgram{n}(\freq)$ for the spectral density matrix function estimate at wavenumber $\freq$ from the $n$th observational window in a sequence of growing windows.

We need some preliminary results before obtaining our required bias results.
This first Lemma is Theorem 3.6 of \cite{andersen1995complex}.

\begin{lemma}
    Consider a complex Wishart distributed random variable $X\sim \mathcal{W}^C_m(n, \Sigma)$.
    Say that $X$ is partitioned so that
    \begin{align}
        X= \begin{bmatrix}
            X_{11} & X_{12} \\ X_{21} & X_{22}
        \end{bmatrix}
    \end{align}
    with $X_{11}$ being $s\times s$ (so $X_{22}$ is $m-s \times m-s$).    
    Then if $X_{11\pred 2} = X_{11}-X_{12}X_{22}^{-1}X_{21}$ we have
    \begin{align}
        X_{11\pred 2} \sim \mathcal{W}^C_s(n-(m-s), \Sigma_{11\pred 2}).
    \end{align}
\end{lemma}

\begin{lemma}\label{lemma:mpt_wishart}
    Under assumptions 1-5 and 7 of \cite{grainger2025spectral}, we have that for a fixed number of processes $P$ and a fixed number of tapers $M$
    \begin{align}
        \mtpgram{n}(\freq) \xrightarrow{d} \mathcal{W}^C_P \left(M, \sdf{}(\freq)/M \right)
    \end{align}
    \begin{proof}
        This follows from Theorem 3 of \cite{grainger2025spectral} and the continuous mapping theorem.
    \end{proof}
\end{lemma}

Now in order to use this result for the expectation, we need to establish uniform integrability conditions.
Strictly speaking, we do not have guarantees that the multitaper periodogram is invertible almost surely.
Therefore, since it is a Gram matrix, and thus positive-semi definite, we can consider instead the generalised schur complement \citep{carlson1974generalization}, defined by
\begin{align}
    \mtpgram{11\pred 2;n}^\dag = \mtpgram{11;n} - \mtpgram{12;n}\mtpgram{22;n}^+\mtpgram{21;n}
\end{align}
where $\mtpgram{22;n}^{\dag}$ is the Moore-Penrose pseudoinverse of $\mtpgram{22;n}$.
This is well-defined almost surely.

\begin{lemma}
    Writing $\norm{\cdot}_F$, say that there exists $\delta>0$ such that $\sup_{n\in\NN} \EE{\norm{\mtpgram{n}}_F^{1+\delta}}<\infty$.
    Then we have
    \begin{align}
        \sup_{n\in\NN} \EE{\norm{\mtpgram{11\pred 2;n}^\dag}_F^{1+\delta}}<\infty.
    \end{align}
    \begin{proof}
        We have
        \begin{align}
            \mtpgram{11;n} - \mtpgram{11\pred 2;n}^\dag
            &= \mtpgram{12;n}\mtpgram{22;n}^{\dag}\mtpgram{21;n} \\
            &= \mtpgram{12;n}\mtpgram{22;n}^{\dag}\mtpgram{12;n}^H\\
            & \succeq 0,
        \end{align}
        because $\mtpgram{22;n}$ is Hermitian and positive semi-definite and therefore so is $\mtpgram{22;n}^{\dag}$ by Theorem 29.6 of \cite{ben2003generalized}.
        Therefore, in Loewner order, $\mtpgram{11;n} \succeq\mtpgram{11\pred 2;n}^\dag$.
        As a result, writing $\norm{\cdot}_F$ for the Frobenius norm (which is monotone on the space of positive semi-definite matrices \citep{ciarlet1989introduction}), we have
        \begin{align}
            \norm{\mtpgram{11\pred 2;n}^\dag}_F \leq \norm{\mtpgram{11;n}}_F \leq \norm{\mtpgram{n}}_F
        \end{align}
        almost surely, and therefore
        \begin{align}
            \sup_{n\in\NN} \EE{\norm{\mtpgram{11\pred 2;n}^\dag}_F^{1+\delta}}\leq \sup_{n\in\NN} \EE{\norm{\mtpgram{n}}_F^{1+\delta}}<\infty,
        \end{align}
        as required.
    \end{proof}
\end{lemma}

\begin{lemma}
    Assuming fourth-order moment conditions already required for \cref{lemma:mpt_wishart}, we have
    \begin{align}
        \sup_{n\in\NN} \EE{\norm{\mtpgram{n}}_F^{2}} < \infty.
    \end{align}
    \begin{proof}
        This follows from the convergence of the mean and variance of the multitaper periodogram to finite values established in \cite{grainger2025spectral} in the proof of Theorem 3.
    \end{proof}
\end{lemma}

\begin{theorem}\label{theorem:bias:partial}
    Given the conditions of the other results in this section, we have

    \begin{align}
        \EE{\mtpgram{11\pred 2; n}^\dag(\freq)}
        &\rightarrow \frac{M+s-P}{M} \sdf{11\pred 2}(\freq)
    \end{align}
    as $n\rightarrow \infty$.
    Note that $P-s=P_Z$ in the notation of the main paper.

    \begin{proof}
        As established, $\mtpgram{n} \xrightarrow{d} \mathcal{W}^C_P \left(M, \sdf{}(\freq)/M \right)$.
        Now we again need to use the continuous mapping theorem for the generalised Schur complement.
        Though the generalised Schur complement is not continuous everywhere, it is continuous almost surely with respect to the Wishart distribution with positive definite covariance (which we have by assumption). Therefore we have $\mtpgram{11\pred 2;n} \xrightarrow{d} \mathcal{W}^C_s(M-(P-s), \Sigma_{11\pred 2}/M)$.
        Furthermore, we also have uniform integrability, and so from standard results \citep[Theorem 25.12]{billingsley2012probability}
        \begin{align}
            \EE{\mtpgram{11\pred 2; n}(\freq)}
            &\rightarrow \EE{A}
        \end{align}
        where $A\sim \mathcal{W}^C_s \left(M-(P-s), \sdf{11\pred 2}(\freq)/M \right)$.
        Finally, by properties of the complex Wishart distribution
        \begin{align}
            \EE{A} &=  (M+s-P) \sdf{11\pred 2}(\freq)/M
        \end{align}
        yielding the desired result.
    \end{proof}
\end{theorem}

\subsection{Fast computation}
The formulae for obtaining partial spectra given above are not the most efficient way to construct them \citep{dahlhaus2000graphical}.
The alternate formulae, from \cite{dahlhaus2000graphical} and \cite{eichler2003partial} continue to hold in this setting, as they are just related to matrix manipulations.
In particular omitting the $\freq$ argument, let $g=f^{-1}$ be the function which is the inverse of $f$ at each wavenumber $\freq$ (not the inverse function), and write $g_{XY}$ to mean the $XY$ element of $g$. Then from \cite{eichler2003partial} we have
\begin{align}
    \sdf{XX\pred (YZ)} & = \frac{1}{g_{XX}},                                                                                    \\
    R_{XY \pred Z}     & = \frac{\sdf{XY\pred Z}}{\sqrt{\sdf{XX\pred Z}\sdf{YY\pred Z}}} = \frac{-g_{XY}}{\sqrt{g_{XX}g_{YY}}}, \\
    \sdf{XY\pred Z}    & = \frac{R_{XY}}{1-\abs{R_{XY}}^2}\sqrt{\sdf{XX\pred (YZ)}\sdf{YY\pred (XZ)}}.
\end{align}
Note there is a difference between $\sdf{XX\pred (YZ)}$ and $\sdf{XX\pred Z}$.
We are interested in $\sdf{XX\pred (YZ)}$ and $\sdf{XY\pred Z}$.
Writing these entirely in terms of $g$, we have
\begin{align}
    \sdf{XY\pred Z}
        & = \frac{R_{XY}}{1-\abs{R_{XY}}^2}\sqrt{\sdf{XX\pred (YZ)}\sdf{YY\pred (XZ)}}                     \\
        & = \frac{-g_{XY}/\sqrt{g_{XX}g_{YY}}}{1-\abs{g_{XY}}^2/g_{XX}g_{YY}}\frac{1}{\sqrt{g_{XX}g_{YY}}} \\
        & = \frac{-g_{XY}}{g_{XX}g_{YY}-\abs{g_{XY}}^2}.
\end{align}
So in summary
\begin{align}
    \sdf{XX\pred (YZ)} & = \frac{1}{g_{XX}},                            \\
    \sdf{XY\pred Z}    & = \frac{-g_{XY}}{g_{XX}g_{YY}-\abs{g_{XY}}^2}.
\end{align}

\subsection{Counterexample to conditional orthogonality} \label{app:counterexample}

Let $N_Z$ be a homogeneous Poisson process with intensity $\intensity{Z}>0$ on $\RR^d$.
Now let $\Lambda(u) = [\int_\RRd g(u-x) N_Z(\de x)]^2$ for some integrable kernel $g:\RR^d\to[0,\infty)$.
Then let $N_X$ and $N_Y$ be Cox processes driven by $\Lambda$ but conditionally independent of each other given $\Lambda$.
Then the partial spectra between $X$ and $Y$ accounting for $Z$ is
\begin{align}
    \sdf{XY\pred Z}(\freq) &= 2\intensity{Z}^2 [\abs{G}^2 \ast \abs{G}^2](\freq), & \freq\in\RRd,\label{eq:counterexample:partialspectra}
\end{align}
where $G$ is the Fourier transform of $g$, $\abs{G}^2$ is its elementwise magnitude square and $\ast$ denotes convolution.
The partial spectra $\sdf{XY\pred Z}(\freq)$ is therefore not identically zero.
This contradicts the claim that tests for zero partial spectra (or partial coherence) can be used as tests of conditional orthogonality, as here the partial spectra is not zero, despite the processes being conditionally independent (and thus conditionally orthogonal).

\begin{proof}[Proof of \cref{eq:counterexample:partialspectra}]
We have
\begin{align}
    \EE{\Lambda(u)\Lambda(0)} 
    &= \EE{\left[\int_\RRd g(u-w) N_Z(\de w)\right]^2 \left[\int_\RRd g(-y) N_Z(\de y)\right]^2} \\
    &= \EE{\int_\RRd \int_\RRd \int_\RRd \int_\RRd g(u-w)g(u-x)  g(-y)g(-z) N_Z(\de w)N_Z(\de x)N_Z(\de y) N_Z(\de z)} \\
    &= \int_\RRd \int_\RRd \int_\RRd \int_\RRd g(u-w)g(u-x)  g(-y)g(-z) \momentmeasure{ZZZZ}(\de w \times \de x \times \de y \times \de z) \\
    %
    %
    &= \lambda^4 \int_\RRd \int_\RRd \int_\RRd \int_\RRd g(u-w)g(u-x)  g(-y)g(-z) \de w \de x \de y \de z  \\
    & \quad + \lambda^3 \int_\RRd \int_\RRd \int_\RRd g(u-w)^2  g(-y)g(-z) \de w \de y \de z \\
    & \quad + 4\lambda^3 \int_\RRd \int_\RRd \int_\RRd g(u-w)g(u-x)  g(-w)g(-z) \de w \de x \de z \\
    & \quad + \lambda^3 \int_\RRd \int_\RRd \int_\RRd g(u-w)g(u-x)  g(-y)^2 \de w \de x \de y \\
    & \quad + \lambda^2 \int_\RRd \int_\RRd g(u-w)^2  g(-y)^2 \de w \de y \\
    & \quad + 2\lambda^2 \int_\RRd \int_\RRd g(u-w)g(u-x)  g(-x)g(-w) \de w \de x \\
    & \quad + 2\lambda^2 \int_\RRd \int_\RRd g(u-w)^2  g(-w)g(-z) \de w \de z \\
    & \quad + 2\lambda^2 \int_\RRd \int_\RRd g(u-w)g(u-x)  g(-w)^2 \de w \de x \\
    & \quad + \lambda \int_\RRd \int_\RRd g(u-w)g(u-w)  g(-w)g(-w) \de w \\
    &= \lambda^4 \norm{g}_1^4
    + 2\lambda^3 \norm{g}_2^2 \norm{g}_1^2
    + 4\lambda^3 \norm{g}_1^2 [g \ast g^\ast](u)
    + \lambda^2 \norm{g}_2^4
    + 2\lambda^2 [g \ast g^\ast](u)^2 \\
    & \quad + 2\lambda^2 \norm{g}_1 [g^2 \ast g^\ast](u)
    + 2\lambda^2 \norm{g}_1 [g^2 \ast g^\ast](-u)
    + \lambda [g^2 \ast (g^\ast)^2](u)
\end{align}
Additionally, we have
\begin{align}
    \EE{\Lambda(u)} & = \EE{\left[\int_\RRd g(u-x) N_Z(\de x)\right]^2} \\
                    & = \lambda^2 \norm{g}_1^2 + \lambda \norm{g}_2^2.
\end{align}
Therefore
\begin{align}
    \reducedcumulantdens{\Lambda\Lambda}(u)
    &= \lambda^4 \norm{g}_1^4
    + 2\lambda^3 \norm{g}_2^2 \norm{g}_1^2
    + 4\lambda^3 \norm{g}_1^2 [g \ast g^\ast](u)
    + \lambda^2 \norm{g}_2^4
    + 2\lambda^2 [g \ast g^\ast](u)^2 \\
    & \quad + 2\lambda^2 \norm{g}_1 [g^2 \ast g^\ast](u)
    + 2\lambda^2 \norm{g}_1 [g^2 \ast g^\ast](-u)
    + \lambda [g^2 \ast (g^\ast)^2](u) \\
    & \quad - \left(\lambda^2 \norm{g}_1^2 + \lambda \norm{g}_2^2\right)^2 \\
    & = 4\lambda^3 \norm{g}_1^2 [g \ast g^\ast](u)
    + 2\lambda^2 \left\{[g \ast g^\ast](u)^2 + \norm{g}_1 [g^2 \ast g^\ast](u) + \norm{g}_1 [g^2 \ast g^\ast](-u)\right\} \\
    &\quad + \lambda [g^2 \ast (g^\ast)^2](u).
\end{align}
As a result, we have
\begin{align}
    \sdf{\Lambda\Lambda}(\freq) & = 4\lambda^3 \norm{g}_1^2 \abs{G(\freq)}^2
    + 2\lambda^2 \left\{[\abs{G}^2 \ast \abs{G}^2](\freq) + 2\norm{g}_1 \Re ([G\ast G](\freq) \conj{G(\freq)}) \right\}
    + \lambda \abs{[G\ast G](\freq)}^2
\end{align}
Now we have
\begin{align}
    \sdf{XX}(\freq) & = \sdf{YY}(\freq) = \sdf{\Lambda\Lambda}(\freq) + \lambda_\Lambda \\
    \sdf{XY}(\freq) & = \sdf{\Lambda\Lambda}(\freq).
\end{align}
Finally,
\begin{align}
    \reducedmomentmeasure{Z\Lambda}(B)
    & = \EE{\int_{\UUd} \Lambda(B+z) N_Z(\de z)} \\
    & = \EE{\int_{\UUd} \int_B \left[\int_{\RRd} g(u-x)N_Z(\de x) \right]^2 \de u N_Z(\de z)} \\
    & = \int_B \int_{\RRd} \int_{\RRd} \int_{\RRd} \indicator_{\UUd}(z) g(u+z-x) g(u+z-y) M_{ZZZ}(\de x\times \de y \times \de z) \\
    &= \int_B \lambda^3 \norm{g}_1^2 + \lambda^2\norm{g}_2^2 + 2\lambda^2 \int_{\RRd} \int_{\RRd} \indicator_{\UUd}(z) g(u) g(u+z-y) \de y \de z \\ 
    & \quad + \lambda \int g(u)^2 \indicator_{\UUd}(z) \de z \de u \\
    &= \int_B \lambda^3 \norm{g}_1^2 + \lambda^2\norm{g}_2^2 + 2\lambda^2 \norm{g}_1 g(u) + \lambda g(u)^2 \de u
\end{align}
Therefore, see $\reducedmomentmeasure{Z\Lambda}(B)$ has density
\begin{align}
    \reducedmomentdens{Z\Lambda}(u) & = \lambda^3 \norm{g}_1^2 + \lambda^2\norm{g}_2^2 + 2\lambda^2 \norm{g}_1 g(u) + \lambda g(u)^2.
\end{align}
So
\begin{align}
    \reducedcumulantdens{Z\Lambda}(u) & = \lambda^3 \norm{g}_1^2 + \lambda^2\norm{g}_2^2 + 2\lambda^2 \norm{g}_1 g(u) + \lambda g(u)^2 - \lambda (\lambda^2 \norm{g}_1^2 + \lambda \norm{g}_2^2) \\
    & = 2\lambda^2 \norm{g}_1 g(u) + \lambda g(u)^2
\end{align}
Therefore
\begin{align}
    \sdf{Z\Lambda}(\freq)
    &= 2\lambda^2 \norm{g}_1 G(\freq) + \lambda [G\ast G](\freq).
\end{align}
Furthermore, $\sdf{ZX}=\sdf{ZY}=\sdf{Z\Lambda}$.

Now we can compute the partial spectra
\begin{align}
    \sdf{XY\pred Z}
    & = \sdf{XY} - \sdf{XZ}\sdf{ZZ}^{-1}\sdf{ZY} \\
    &= \sdf{XY} - \lambda^{-1} \conj{\sdf{ZX}}\sdf{ZY} \\
    &= \sdf{XY} - \lambda^{-1} \abs{\sdf{Z\Lambda}}
\end{align}
so for all $\freq\in\RRd$
\begin{align}
    \sdf{XY\pred Z}(\freq) 
    &= \sdf{XY} - \abs{2\lambda^2 \norm{g}_1 G(\freq) + \lambda [G\ast G](\freq)}^2/\lambda \\
    &= \sdf{XY} - 4\lambda^3 \norm{g}_1^2 \abs{G(\freq)}^2 + 2\lambda^2 \Re ([G\ast G](\freq) \conj{G(\freq)}) + \lambda \abs{[G\ast G](\freq)}^2 \\
    &= 2\lambda^2 [\abs{G}^2 \ast \abs{G}^2](\freq)
\end{align}

This is not equal to zero everywhere (unless $g=0$ almost everywhere, which corresponds to having no points of type $X$ or $Y$ almost surely).

\end{proof}

\section{Resampling}\label{sec:resampling}

One common approach to inference for point processes is to compute null envelopes for the statistic of interest, see \citep{myllymaki2017global} for example.
This typically requires simulation under the null of choice.
In the univariate case, the null is usually that the process is a Poisson process, and so we simply simulate from a Poisson process with an intensity equal to the estimated intensity from the observed point pattern.
In the bivariate case, \cite{mrkvivcka2021revisiting} review some common methods for performing null resampling, when the null hypothesis is that the two processes are independent.
The essence of the resampling is to shift one of the two patterns randomly relative to the other, breaking their cross dependence, but retaining marginal properties.
Various schemes are designed to deal with boundary problems which arise from this shifting \citep{mrkvivcka2021revisiting}.

The null in our case is more complicated as it is not a pair of processes which we claim are uncorrelated, but rather the residual processes.
Ideally, if we have realisations of the residuals, we could shift those using a similar strategy to those described by \cite{mrkvivcka2021revisiting}.
However, we do not have direct access to these residuals.
Such resampling in our case is tricky, and remains an open problem, however, we will briefly discuss some of the approaches which we have considered.

Firstly let us consider the intraprocess case.
In this case, in fact, the solution is theoretically straightforward, but practically difficult.
An analogous null to the standard case if we are interested in the partial $K$ function between $X$ and itself accounting for $Z$, is to simulate from an inhomogeneous Poisson process with intensity given by $\Predprocess{X\pred Z}$.
In practice then, we could replace $\Predprocess{X\pred Z}$ with some form of estimate.
However, this approach is sensitive in practice.

In the interprocess case, the problem is even more difficult.
Say that we are interested in the partial $K$ function between $X$ and $Y$ accounting for $Z$.
Then the null hypothesis is that the residual process of $X$ with $Z$ and the residual process of $Y$ with $Z$ are uncorrelated.
There are three approaches we could immediately consider here.
The first is to apply the standard shifting shifting approaches \citep{mrkvivcka2021revisiting} to all of the processes involved, and then compute the partial $K$ function between the shifted processes.
This has the advantage of being simple to implement, but has the disadvantage that we make all of the different processes independent, rather than just the residuals.
The second approach would be to shift only the $X$ process say, breaking the dependence between the residual processes, but also breaking some of the other dependences.
The third approach would be to again simulate from inhomogeneous Poisson processes with intensities given by $\Predprocess{X\pred Z}$ and $\Predprocess{Y\pred Z}$ respectively.

Say that we are in the cross-process case, and we have processes of interest $X$ and $Y$, and covariate processes $Z$.
Then there are essentially (due to symmetries) 6 quantities of interest: $\sdf{XX}, \; \sdf{YY}, \; \sdf{ZZ}, \; \sdf{XY}, \; \sdf{XZ}, \; \sdf{YZ}$.
Now, the partial spectra of interest is $\sdf{XY\pred Z} = \sdf{XY} - \sdf{XZ}\sdf{ZZ}^{-1}\sdf{ZY}$.
We want the resampling method to set $\sdf{XY\pred Z}$ to zero, whilst maintaining $\sdf{XX}$, $\sdf{YY}$, $\sdf{XZ}$, $\sdf{YZ}$ and $\sdf{ZZ}$.
In other words, by setting $\sdf{XY} = \sdf{XZ}\sdf{ZZ}^{-1}\sdf{ZY}$.
However, none of the methods we have previously discussed achieve this exactly.
\Cref{tab:resampling_methods} summarises the effects of the different resampling methods on the different components of the spectral density function.
In pilot simulation studies, we have found that these approaches can perform well in some scenarios, but when there is strong dependence between the processes, they can fail to achieve the desired Type I error rate, in some cases catastrophically (we saw rejection rates around 40\% when the nominal level was 5\%).
Therefore, this aspect requires further research.

\begin{table}[h]
\centering
\begin{tabular}{l|cccc}
\toprule
\textbf{Component} & \textbf{Ideal Method} & \textbf{Shift all} & \textbf{Shift $X$} & \textbf{Cox Generation} \\
\midrule
$\sdf{XX}$ & --- & --- & --- & $\sdf{XZ}\sdf{ZZ}^{-1}\sdf{ZX} + \lambda_X$ \\
$\sdf{YY}$ & --- & --- & --- & $\sdf{YZ}\sdf{ZZ}^{-1}\sdf{ZY} + \lambda_Y$ \\
$\sdf{ZZ}$ & --- & --- & --- & --- \\
$\sdf{XY}$ & $\sdf{XZ}\sdf{ZZ}^{-1}\sdf{ZY}$ & $0$ & $0$ & $\sdf{XZ}\sdf{ZZ}^{-1}\sdf{ZY}$ \\
$\sdf{XZ}$ & --- & $0$ & $0$ & --- \\
$\sdf{YZ}$ & --- & $0$ & --- & --- \\
\bottomrule
\end{tabular}
\caption{Effects of different resampling methods on spectral density components. A --- indicates that a component is preserved.}
\label{tab:resampling_methods}
\end{table}

\section{Example models}
\label{app:example}
\subsection{Model definition}
Consider the following simple model.
Say that $Z$ is a homogeneous Poisson process with intensity $\lambda_Z$.
Say that $X$ and $Y$ are generated by cluster processes, which cluster independently around $Z$ points, with $\clusternumber{X}$ points per $X$ cluster and $\clusternumber{Y}$ points per $Y$ cluster.
Say that the distribution of the difference of the point from its parent are given by $\clusternumber{X}$ and $\clusternumber{Y}$ for points of type $X$ and $Y$ respectively.
If they have densities, we will write $\clusterdens{X}$ and $\clusterdens{Y}$.

\subsection{Model properties}
Then one can rewrite the cluster process as a Cox process with the driving intensity measure
\begin{equation}
    \Lambda_X(B) = \clusteraverage{X} \int_\RRd \clusteroffset{X}(B-u) N_Z(\de u),
\end{equation}
and similarly for $N_Y$ \citep{daley2003introduction}.

\begin{lemma}\label{lemma:example:covariance}
    We have the following reduced moment measures
    \begin{align}
        \reducedmomentmeasure{XZ}(B) & = \clusteraverage{X}\int_\RRd \clusteroffset{X}(B-u) \reducedmomentmeasure{ZZ}(\de u),                                                                                       \\
        \reducedmomentmeasure{XY}(B) & = \clusteraverage{X} \clusteraverage{Y} \int_{\RRd} \int_\RRd \clusteroffset{X}(B+y'-u) \clusteroffset{Y}(\de y') \reducedmomentmeasure{ZZ}(\de u),                          \\
        \reducedmomentmeasure{XX}(B) & =\clusteraverage{X}\intensity{Z}\delta(B) + \clusteraverage{X}^2 \int_{\RRd} \int_\RRd \clusteroffset{X}(B+y'-u) \clusteroffset{X}(\de y') \reducedmomentmeasure{ZZ}(\de u),
    \end{align}
    where $\delta$ is the Dirac measure.
\end{lemma}

\begin{lemma}
    If $\clusteroffset{X}$ and $\clusteroffset{Y}$ admit densities $\clusterdens{X}$ and $\clusterdens{Y}$, then the reduced moment measures admit densities (except for the atom at zero), so that
    \begin{align}
        \reducedmomentdens{XZ}(x) & = \clusteraverage{X} \int_\RRd \clusterdens{X}(x-u) \reducedmomentmeasure{ZZ}(\de u),                                                                                            \\
        \reducedmomentdens{XY}(x) & = \clusteraverage{X} \clusteraverage{Y} \int_{\RRd} \int_\RRd \clusterdens{X}(x+y'-u) \clusterdens{Y}(y') \leb{\de y'} \reducedmomentmeasure{ZZ}(\de u),                         \\
        \reducedmomentdens{XX}(x) & =\clusteraverage{X}\intensity{Z}\delta(x) + \clusteraverage{X}^2 \int_{\RRd} \int_\RRd \clusterdens{X}(x+y'-u) \clusterdens{X}(y') \leb{\de y'} \reducedmomentmeasure{ZZ}(\de u)
    \end{align}
    where here $\delta$ means the Dirac delta function.
\end{lemma}

For the distribution $\clusteroffset{X}$, we will write
\begin{equation}
    \clusterchar{X}(\freq) = \int_\RRd e^{2\pi i \ip{x}{\freq}} \clusteroffset{X}(\de x)
\end{equation}
to be the characteristic function (the inverse Fourier transform of the distribution).
Note we are including the factor of $2\pi$ in our convention.

\begin{lemma}\label{lemma:example:spectra}
    We have the following spectral density functions (with the remaining coming from relabelling and symmetry)
    \begin{align}
        \sdf{XZ}(\freq) & = \clusteraverage{X}\conj{\clusterchar{X}(\freq)} \sdf{ZZ}(\freq),                                            \\
        \sdf{XY}(\freq) & = \clusteraverage{X} \clusteraverage{Y} \conj{\clusterchar{X}(\freq)} \clusterchar{Y}(\freq) \sdf{ZZ}(\freq), \\
        \sdf{XX}(\freq) & = \clusteraverage{X}\intensity{Z} + \clusteraverage{X}^2 \abs{\clusterchar{X}(\freq)}^2 \sdf{ZZ}(\freq).
    \end{align}
\end{lemma}

\begin{proposition}\label[proposition]{prop:pred:kernel}
    The prediction kernel for predicting $X$ from $Z$ is given by
    \begin{equation}
        \Predkernel{X\pred Z} = \clusteraverage{X}\clusteroffset{X}.
    \end{equation}
\end{proposition}

\begin{proposition}\label[proposition]{prop:example:partialspectra}
    The partial spectra for the cluster model are
    \begin{align}
        \sdf{XY\pred Z}(\freq)     & = 0,                                                                                                                                                                                    \\
        \sdf{X,Z \pred Y}(\freq)   & = \frac{\clusteraverage{X} \intensity{Z} \conj{\clusterchar{X}(\freq)} \sdf{ZZ}(\freq)}{\intensity{Z} + \clusteraverage{Y} \abs{\clusterchar{Y}(\freq)}^2 \sdf{ZZ}(\freq)},             \\
        \sdf{X,X\pred Y, Z}(\freq) & = \intensity{X},                                                                                                                                                                        \\
        \sdf{Z,Z\pred X, Y}(\freq) & = \frac{\intensity{Z}\sdf{ZZ}(\freq)}{\clusteraverage{Y}\abs{\clusterchar{Y}(\freq)}^2\sdf{ZZ}(\freq) + \clusteraverage{X}\abs{\clusterchar{X}(\freq)}^2\sdf{ZZ}(\freq)+\intensity{Z}}.
    \end{align}
\end{proposition}

\subsection{Proofs of model properties}
Let $\momentmeasure{XY}$ be the joint moment measure of $N_X$ and $N_Y$ so that $\momentmeasure{XY}(A\times B) = \EE{N_X(A) N_Y(B)}$ for $A,B\in\borel{\RR^d}$.
This relates to the reduced moment measure by
\begin{align}
    \int_{\RR^{2d}} g(x, y) \momentmeasure{XY}(\de x\times \de y) = \int_\RRd \int_\RRd g(z+u, z) \reducedmomentmeasure{XY}(\de u) \leb{\de z}
\end{align}
for bounded measurable functions $g:\RR^{2d}\to\RR$ of bounded support.

\begin{proof}[Proof of \cref{lemma:example:covariance}]
    From \cite{daley2003introduction} Chapter 6, we have
    \begin{align}
        \intensity{X}                & = \clusteraverage{X} \intensity{Z},                                    \\
        \reducedmomentmeasure{XZ}(B) & = \reducedmomentmeasure{\Lambda_X,Z}(B),                               \\
        \reducedmomentmeasure{XY}(B) & = \reducedmomentmeasure{\Lambda_X,\Lambda_Y}(B),                       \\
        \reducedmomentmeasure{XX}(B) & = \reducedmomentmeasure{\Lambda_X,\Lambda_X}(B) + \lambda_X \delta(B).
    \end{align}
    The remaining can be obtained by relabelling and symmetry.
    Proceeding in turn, firstly
    \begin{align}
        \reducedmomentmeasure{XZ}(B)
            & = \reducedmomentmeasure{\Lambda_X,Z}(B)                                                                                          \\
            & = \EE{\int_\UUd \Lambda_X(B+y) N_Z(\de y)}                                                                                       \\
            & = \clusteraverage{X} \EE{\int_\RRd \int_\RRd \clusteroffset{X}(B+y-x) \indicator_{\UUd}(y) N_Z(\de x) N_Z(\de y)}                \\
            & = \clusteraverage{X}\int_{\RR^{2d}} \clusteroffset{X}(B+y-x) \indicator_{\UUd}(y) \momentmeasure{ZZ}(\de x \times \de y)         \\
            & = \clusteraverage{X}\int_\RRd \int_\RRd \clusteroffset{X}(B-u) \indicator_{\UUd}(z) \reducedmomentmeasure{ZZ}(\de u) \leb{\de z} \\
            & = \clusteraverage{X}\int_\RRd \clusteroffset{X}(B-u) \reducedmomentmeasure{ZZ}(\de u).
    \end{align}
    Secondly
    \begin{align}
        \reducedmomentmeasure{XY}(B)
            & = \reducedmomentmeasure{\Lambda_X,\Lambda_Y}(B)                                                                                                                                                   \\
            & = \EE{\int_\UUd \Lambda_X(B+y) \Lambda_Y(\de y)}                                                                                                                                                  \\
            & = \clusteraverage{X} \clusteraverage{Y}\EE{\int_\RRd \int_\RRd \int_\RRd \clusteroffset{X}(B+y-x) \indicator_{\UUd}(y) N_Z(\de x) \clusteroffset{Y}(\de y-z) N_Z(\de z)}                          \\
            & = \clusteraverage{X} \clusteraverage{Y}\EE{\int_\RRd \int_\RRd \int_\RRd \clusteroffset{X}(B+y'+z-x) \indicator_{\UUd}(y'+z) N_Z(\de x) \clusteroffset{Y}(\de y') N_Z(\de z)}                     \\
            & = \clusteraverage{X} \clusteraverage{Y} \int_{\RR^{2d}} \int_\RRd \clusteroffset{X}(B+y'+z-x) \indicator_{\UUd}(y'+z)  \clusteroffset{Y}(\de y') \momentmeasure{ZZ}(\de x \times \de z)           \\
            & = \clusteraverage{X} \clusteraverage{Y} \int_{\RRd} \int_\RRd \int_\RRd \clusteroffset{X}(B+y'-u) \indicator_{\UUd}(y'+z)  \clusteroffset{Y}(\de y') \reducedmomentmeasure{ZZ}(\de u) \leb{\de z} \\
            & = \clusteraverage{X} \clusteraverage{Y} \int_{\RRd} \int_\RRd \int_\RRd \clusteroffset{X}(B+y'-u) \indicator_{\UUd}(z')  \leb{\de z'} \clusteroffset{Y}(\de y') \reducedmomentmeasure{ZZ}(\de u)  \\
            & = \clusteraverage{X} \clusteraverage{Y} \int_{\RRd} \int_\RRd \clusteroffset{X}(B+y'-u) \clusteroffset{Y}(\de y') \reducedmomentmeasure{ZZ}(\de u).
    \end{align}
    Finally by a similar argument
    \begin{align}
        \reducedmomentmeasure{XX}(B) & =\clusteraverage{X}\intensity{Z}\delta(B) + \clusteraverage{X}^2 \int_{\RRd} \int_\RRd \clusteroffset{X}(B+y'-u) \clusteroffset{X}(\de y') \reducedmomentmeasure{ZZ}(\de u).
    \end{align}
\end{proof}

\begin{proof}[Proof of \cref{lemma:example:spectra}]
    Note that we have immediately that
    \begin{align}
        \lambda_X \ell(B) &= \clusteraverage{X}\intensity{Z}\ell(B) \\
        &= \clusteraverage{X}\intensity{Z}\int_{\RRd} \indicator_B(u) \ell(\de u) \int_\RRd \clusteroffset{X}(\de x) \\
        &= \clusteraverage{X}\intensity{Z}\int_{\RRd} \int_\RRd \indicator_B(u+x)   \clusteroffset{X}(\de x) \ell(\de u) \\
        &= \clusteraverage{X}\intensity{Z}\int_{\RRd} \clusteroffset{X}(B-u) \ell(\de u), \\
    \end{align}
    and therefore
    \begin{align}
        \reducedcumulantmeasure{XZ}(B) & = 
        \reducedmomentmeasure{XZ}(B) - \intensity{X}\intensity{Z}\ell(B)\\
        & = \clusteraverage{X}\int_\RRd \clusteroffset{X}(B-u) \reducedmomentmeasure{ZZ}(\de u) -  \clusteraverage{X}\intensity{Z}^2\int_{\RRd} \clusteroffset{X}(B-u) \ell(\de u) \\
        & = \clusteraverage{X}\int_\RRd \clusteroffset{X}(B-u) \reducedcumulantmeasure{ZZ}(\de u),  \\
    \end{align}

    Similarly we have
    \begin{align}
        \reducedcumulantmeasure{XY}(B) & = \clusteraverage{X} \clusteraverage{Y} \int_{\RRd} \int_\RRd \clusteroffset{X}(B+y'-u) \clusteroffset{Y}(\de y') \reducedcumulantmeasure{ZZ}(\de u),                          \\
        \reducedcumulantmeasure{XX}(B) & =\clusteraverage{X}\intensity{Z}\delta(B) + \clusteraverage{X}^2 \int_{\RRd} \int_\RRd \clusteroffset{X}(B+y'-u) \clusteroffset{X}(\de y') \reducedcumulantmeasure{ZZ}(\de u).
    \end{align}

    Therefore, firstly
    \begin{align}
        \sdf{XZ}(\freq)
            & = \int_\RRd e^{-2\pi i \ip{x}{\freq}} \reducedcumulantmeasure{XZ}(\de x)                                                             \\
            & = \int_\RRd e^{-2\pi i \ip{x}{\freq}} \clusteraverage{X}\int_\RRd \clusteroffset{X}(\de x-u) \reducedcumulantmeasure{ZZ}(\de u)      \\
            & = \clusteraverage{X} \int_\RRd \int_\RRd e^{-2\pi i \ip{(x'+u)}{\freq}} \clusteroffset{X}(\de x') \reducedcumulantmeasure{ZZ}(\de u) \\
            & = \clusteraverage{X}\clusterchar{X}(-\freq) \sdf{ZZ}(\freq)                                                                        \\
            & = \clusteraverage{X}\conj{\clusterchar{X}(\freq)} \sdf{ZZ}(\freq).
    \end{align}
    Secondly,
    \begin{align}
        \sdf{XY}(\freq)
            & = \int_\RRd e^{-2\pi i \ip{x}{\freq}} \reducedcumulantmeasure{XY}(\de x)                                                                                                                     \\
            & = \int_\RRd e^{-2\pi i \ip{x}{\freq}} \clusteraverage{X} \clusteraverage{Y} \int_{\RRd} \int_\RRd \clusteroffset{X}(\de x+y'-u) \clusteroffset{Y}(\de y') \reducedcumulantmeasure{ZZ}(\de u) \\
            & = \clusteraverage{X} \clusteraverage{Y} \int_\RRd \int_\RRd e^{-2\pi i \ip{(x'-y'+u)}{\freq}} \clusteroffset{X}(\de x') \clusteroffset{Y}(\de y') \reducedcumulantmeasure{ZZ}(\de u)         \\
            & = \clusteraverage{X} \clusteraverage{Y} \conj{\clusterchar{X}(\freq)} \clusterchar{Y}(\freq) \sdf{ZZ}(\freq).
    \end{align}
    Finally, again by the same logic
    \begin{align}
        \sdf{XX}(\freq)
            & = \clusteraverage{X}\intensity{Z} + \clusteraverage{X}^2 \abs{\clusterchar{X}(\freq)}^2 \sdf{ZZ}(\freq)
    \end{align}
\end{proof}

\begin{proof}[Proof of \cref{prop:pred:kernel}]
    Firstly, we note that
    \begin{align*}
        \predkernelft{X\pred Z}(\freq) &= \sdf{XZ}(\freq)\sdf{ZZ}(\freq)^{-1} \\
        & = \clusteraverage{X}\conj{\clusterchar{X}(\freq)} \sdf{ZZ}(\freq) \sdf{ZZ}(\freq)^{-1} \\
        & = \clusteraverage{X}\conj{\clusterchar{X}(\freq)}.
    \end{align*}
    Recalling that characteristic functions are the inverse Fourier transform of the distribution, we see that $\conj{\clusterchar{X}(\freq)}$ is the usual Fourier transform of $\clusteroffset{X}$ (by conjugate symmetry), and therefore $\predkernel{X\pred Z} = \clusteraverage{X}\clusteroffset{X}$.
\end{proof}

\begin{proof}[Proof of \cref{prop:example:partialspectra}]
    Using \cref{lemma:example:spectra}, note a couple of useful relationships
    \begin{align}
        \sdf{XX}(\freq)         & = \intensity{X} + \abs{\sdf{XZ}(\freq)}^2/\sdf{ZZ}(\freq)        \\
        \sdf{XY}(\freq)         & = \sdf{XZ}(\freq) \sdf{ZY}(\freq)/\sdf{ZZ}(\freq)                \\
        \abs{\sdf{XY}(\freq)}^2 & = (\sdf{XX}(\freq)-\intensity{X})(\sdf{YY}(\freq)-\intensity{Y})
    \end{align}

    Therefore
    \begin{align}
        \sdf{XY\pred Z}(\freq)
            & = \sdf{XY}(\freq) - \sdf{XZ}(\freq) \sdf{ZZ}(\freq)^{-1} \sdf{ZY}(\freq) \\
            & = \sdf{XY}(\freq) - \sdf{XY}(\freq)                                      \\
            & = 0.
    \end{align}

    Next, consider
    \begin{align}
        \sdf{X,Z \pred Y}(\freq)
            & = \sdf{XZ}(\freq) - \sdf{XY}(\freq) \sdf{YY}(\freq)^{-1} \sdf{Y,Z}(\freq)                                                                                                                                                                                                                                                                       \\
            & = \clusteraverage{X}\conj{\clusterchar{X}(\freq)} \sdf{ZZ}(\freq) - \frac{\clusteraverage{X} \clusteraverage{Y} \conj{\clusterchar{X}(\freq)} \clusterchar{Y}(\freq) \clusteraverage{Y} \conj{\clusterchar{Y}(\freq)} \sdf{ZZ}(\freq)^2}{\clusteraverage{Y}\intensity{Z} + \clusteraverage{Y}^2 \abs{\clusterchar{Y}(\freq)}^2 \sdf{ZZ}(\freq)} \\
            & = \clusteraverage{X}\conj{\clusterchar{X}(\freq)} \sdf{ZZ}(\freq) - \frac{\clusteraverage{X} \clusteraverage{Y} \conj{\clusterchar{X}(\freq)} \abs{\clusterchar{Y}(\freq)}^2 \sdf{ZZ}(\freq)^2}{\intensity{Z} + \clusteraverage{Y} \abs{\clusterchar{Y}(\freq)}^2 \sdf{ZZ}(\freq)}                                                              \\
            & = \frac{\clusteraverage{X} \intensity{Z} \conj{\clusterchar{X}(\freq)} \sdf{ZZ}(\freq)}{\intensity{Z} + \clusteraverage{Y} \abs{\clusterchar{Y}(\freq)}^2 \sdf{ZZ}(\freq)}
    \end{align}

    Now letting $W=(Y,Z)^\top$, we have omitting the ``$(\freq)$'' for clarity
    \begin{align}
        \sdf{X,X\pred W}
            & = \sdf{XX} - \sdf{X,W} \sdf{W,W}^{-1} \sdf{W,X}                                                                                                                                                    \\
            & = \sdf{XX} - \begin{bmatrix} \sdf{XY} & \sdf{XZ} \end{bmatrix} \begin{bmatrix} \sdf{YY} & \sdf{Y,Z} \\ \sdf{ZY} & \sdf{ZZ} \end{bmatrix}^{-1} \begin{bmatrix} \sdf{Y,X} \\ \sdf{Z,X} \end{bmatrix} \\
            & = \sdf{XX} - \frac{\abs{\sdf{XY}}^2\sdf{ZZ} + \abs{\sdf{XZ}}^2\sdf{YY} - 2\Re{\sdf{XZ}\sdf{ZY}\sdf{Y,X}}}{\sdf{YY}\sdf{ZZ} - \abs{\sdf{Y,Z}}^2}                                                    \\
            & = \sdf{XX} - \frac{\abs{\sdf{XY}}^2\sdf{ZZ} + \abs{\sdf{XZ}}^2\sdf{YY} - 2\abs{\sdf{XY}}^2\sdf{ZZ}}{\intensity{Y}\sdf{ZZ}}.
    \end{align}
    Now
    \begin{align}
        \abs{\sdf{XZ}}^2\sdf{YY} & = \abs{\sdf{XZ}}^2\intensity{Y} + \abs{\sdf{XZ}}^2\abs{\sdf{Y,Z}}^2 / \sdf{ZZ} \\
                                    & = \abs{\sdf{XZ}}^2\intensity{Y} + \abs{\sdf{XY}}^2\sdf{ZZ}
    \end{align}
    and therefore
    \begin{align}
        \sdf{X,X\pred W}
            & = \sdf{XX} - \frac{\abs{\sdf{XZ}}^2}{\sdf{ZZ}} \\
            & = \intensity{X}.
    \end{align}
    Finally, we have letting $U = (Y,X)^\top$
    \begin{align}
        \sdf{Z,Z\pred U}
            & = \sdf{ZZ} - \frac{\abs{\sdf{ZY}}^2\sdf{XX} + \abs{\sdf{Z,X}}^2\sdf{YY} - 2\Re{\sdf{Z,X}\sdf{XY}\sdf{Y,Z}}}{\sdf{YY}\sdf{XX} - \abs{\sdf{Y,X}}^2}                                                                                                                                                                                    \\
            & = \sdf{ZZ} - \frac{2\sdf{XX}\sdf{YY}\sdf{ZZ} - \intensity{X}\sdf{YY}\sdf{ZZ}-\intensity{Y}\sdf{XX}\sdf{ZZ}- 2\abs{\sdf{XY}}^2\sdf{ZZ}}{\sdf{YY}\sdf{XX} - \abs{\sdf{Y,X}}^2}                                                                                                                                                         \\
            & = \frac{\intensity{X}\sdf{YY}\sdf{ZZ}+\intensity{Y}\sdf{XX}\sdf{ZZ}}{\sdf{YY}\sdf{XX} - \abs{\sdf{Y,X}}^2} - \sdf{ZZ}                                                                                                                                                                                                                \\
            & = \sdf{ZZ}\frac{\intensity{X}\sdf{YY}+\intensity{Y}\sdf{XX}-\sdf{YY}\sdf{XX} + \abs{\sdf{Y,X}}^2}{\sdf{YY}\sdf{XX} - \abs{\sdf{Y,X}}^2}                                                                                                                                                                                              \\
            & = \sdf{ZZ}\frac{\intensity{X}\sdf{YY}+\intensity{Y}\sdf{XX}-\sdf{YY}\sdf{XX} + (\sdf{XX}-\intensity{X})(\sdf{YY}-\intensity{Y})}{\sdf{YY}\sdf{XX} - \abs{\sdf{Y,X}}^2}                                                                                                                                                               \\
            & = \sdf{ZZ}\frac{\intensity{X}\intensity{Y}}{\sdf{YY}\sdf{XX} - \abs{\sdf{Y,X}}^2}                                                                                                                                                                                                                                                    \\
            & = \sdf{ZZ}\frac{\intensity{X}\intensity{Y}}{\intensity{X}\sdf{YY}+\intensity{Y}\sdf{XX} - \intensity{X}\intensity{Y}}                                                                                                                                                                                                                \\
            & = \frac{\sdf{ZZ}\intensity{Z}^2\clusteraverage{X}\clusteraverage{Y}}{\intensity{Z}\clusteraverage{X}\clusteraverage{Y}^2\abs{\clusterchar{Y}}^2\sdf{ZZ} + \intensity{X}\intensity{Y} + \intensity{Z}\clusteraverage{Y}\clusteraverage{X}^2\abs{\clusterchar{X}}^2\sdf{ZZ} + \intensity{X}\intensity{Y} - \intensity{X}\intensity{Y}} \\
            & = \frac{\sdf{ZZ}\intensity{Z}}{\clusteraverage{Y}\abs{\clusterchar{Y}}^2\sdf{ZZ} + \clusteraverage{X}\abs{\clusterchar{X}}^2\sdf{ZZ}+\intensity{Z}}
    \end{align}
\end{proof}

\section{Simulation study details}\label{app:simulation_details}

\subsection{Bivariate examples}
The first model was constructed by $Y$ points being a Poisson process, with a child processes of type $X_0$ having an average of $\clusternumber{X_0}$ with the offset between children and the parent, say $D_{X_0}$, following a multivariate Gaussian distribution.
In particular, $\intensity{Y}=0.01$, $\clusteraverage{X_0} = 3$, and $D_{X_0} \sim \mathcal{N}(0, 1.5^2 I)$.
We then label $X=X_0$.

To generate the second model, we start with the same setup as the previous model (for $Y$ and $X_0$).
We then generate the new $X$ points be again taking clusters, now centred around points of type $X_0$, with $\clusteraverage{X} = 1$, and the additional offset $D_{X} \sim \mathcal{N}(0, I)$. 
The parameters for $Y$ and $X_0$ are kept the same as the first example here.

Finally, for the third scenario, we again start with a similar initial $Y,X_0$ setup, and then $X$ is set to be the survivors of a thinning process.
In particular, each point of type $X_0$ is assigned a mark uniformly on $[0,1]$.
Then if two points are within a distance $r_X$ of each other, the point with the smaller mark is removed with probability $1-p_X$.
For the specific example, we use $\clusteraverage{X_0} = 15$, and $D_{X_0} \sim \mathcal{N}(0, 1.5^2 I)$, $r_X = 3$ and $p_X = 0.1$.
The parameters for $Y$ are the same as for the first example.
The resulting $X$ process is marginally a generalisation of a Mat\'ern hard-core process, though not quite the same as those in \cite{teichmann2013generalizations}.

\subsection{Trivariate examples}

In the trivariate case, we start with a homogeneous Poisson process $Z$.
In all three cases, we set $\intensity{Z}=0.01$.

For the first setting, we generate points of type $X$ as described in the previous section, with $\clusteraverage{X} = 3$ and $D_{X} \sim \mathcal{N}(0, 2^2 I)$. Points of type $Y$ are generated in the same way, but independently of points of type $X$.

For the second setting, we generate points of type $Z$ and $Y$ as just described.
Points of type $X$ were then generated by placing clusters around points of type $Y$ with with $\clusteraverage{X} = 1$ and $D_{X} \sim \mathcal{N}(0, 2^2 I)$.

In the final case, points of type $Z$ and $Y$ are again generated as in the first example.
We then generate points of type $X_0$ centred at points of type $Z$ with $\clusteraverage{X_0} = 10$ and $D_{X_0} \sim \mathcal{N}(0, 2^2 I)$.
We then conditionally thin them based on their distance to points of type $Y$, with survival probability $p_X=0.1$ and inhibition radius $r_X=3$.

\section{Additional simulation results}
\subsection{Example accounting for covariates}\label{app:simulation:inhom}

\Cref{fig:trivariateexampleinhom} shows an example of a trivariate system where the $Z$ process is inhomogeneous, with intensity $\lambda_Z([z_1,z_2]) = 0.02z_1/100$ over the region $[0,100]^2$.
We then consider the partial $L$ functions between each pair of processes, accounting for the other process and the intensity (which we treat as known for both the $L$ and partial $L$ functions.
For example, when considering $X$ vs $Y$, we account for the point pattern $Z$ and the covariate $\lambda_Z$.\footnote{Strictly speaking, this violates the assumption that the processes are jointly homogeneous. However, the partial L function does recover the structure in this case, and the relationships between processes are still the invariant to shifts, even though the intensity is deterministic.}
Again we plot the three different interaction scenarios.
We see that even in this inhomogeneous setting, we are still able to recover the underlying structure using the partial $L$ function.
In particular, the inhomogeneous $L$ function discovers clustering in almost all cases (except between $Y$ and $Z$ in the agnostic offspring case, where there should be clustering).
In contrast, the partial $L$ function correctly recovers the same structure as that seen in the corresponding example in the main manuscript.

\begin{figure}[hpt!]
    \centering
    \includegraphics{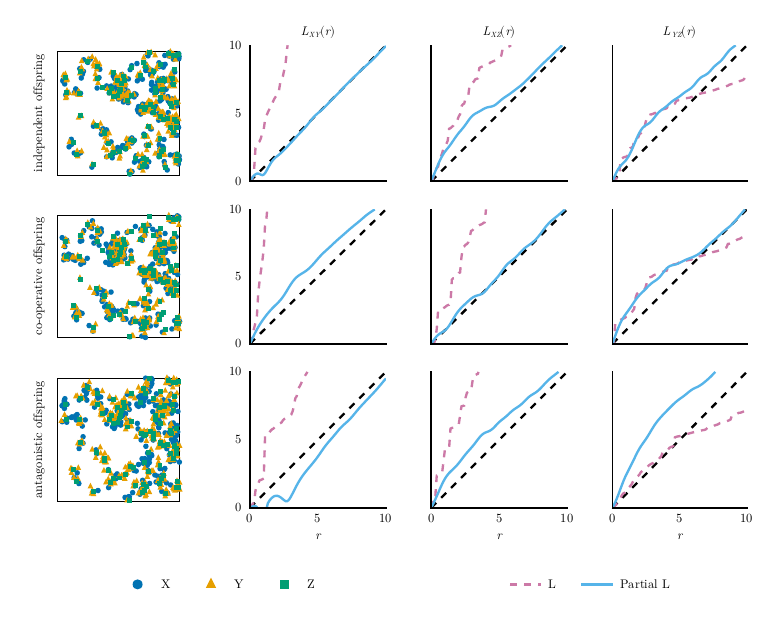}
    \caption{An inhomogeneous version of the trivariate examples. We again show the same trivariate system with three different interaction types, but now the intensity of the base process $Z$ is not constant. The first column shows example processes, the second column shows the $L$ function and partial $L$ function between process $X$ and process $Y$ (accounting for $Z$ and $\lambda_Z$), the third column shows interactions between $X$ and $Z$ (accounting for $Y$ and $\lambda_Z$) and the final column shows interactions between $Y$ and $Z$ (accounting for $X$ and $\lambda_Z$).}
    \label{fig:trivariateexampleinhom}
\end{figure}

\subsection{Pair correlation functions from examples}\label{app:simulation:pcf}

\Cref{fig:predator_prey_example:pcf,fig:trivariateexample:pcf} show the pair correlation functions corresponding the the example processes in the main manuscript.

\begin{figure}[h]
    \centering
    \includegraphics[width=1\textwidth]{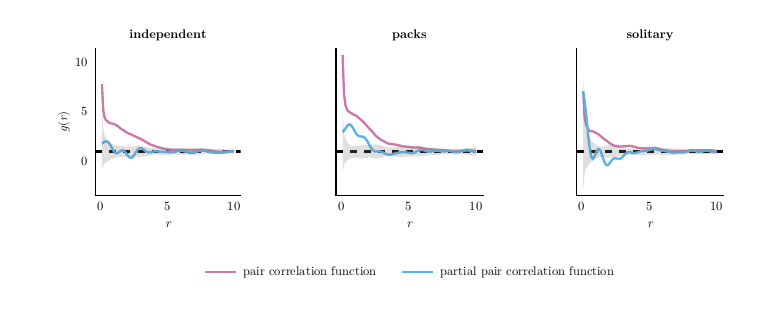}
    \caption{The pair correlation function and partial pair correlation function between the predator process ($X$) in the latter case accounting for the prey process ($Y$). The left column shows the first scenario, where the predators do not interact with each other. The middle column shows the second scenario, where the predators hunt in packs. The right column shows the third scenario, where the predators do not like to be near each other. The envelopes are 95\% confidence envelopes using the MAD envelopes proposed by \cite{myllymaki2017global}.}
    \label{fig:predator_prey_example:pcf}
\end{figure}

\begin{figure}[hpt!]
    \centering
    \includegraphics{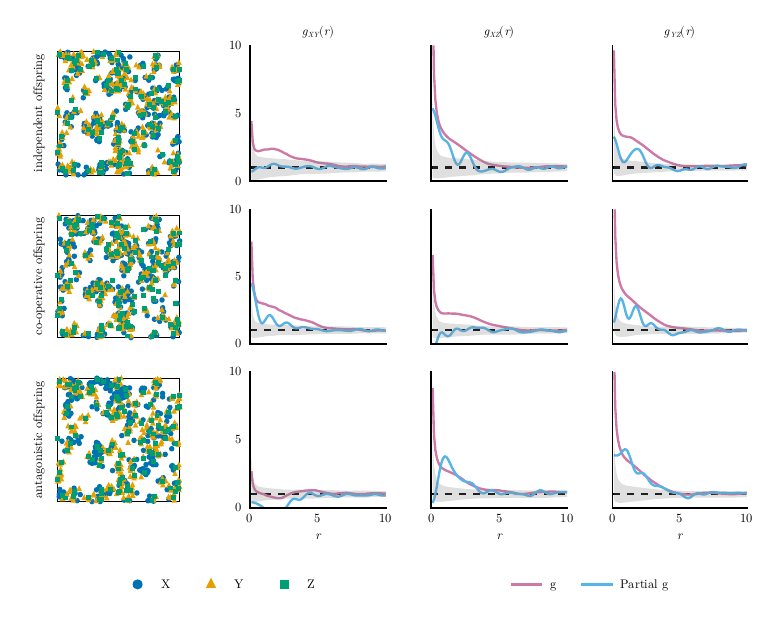}
    \caption{Example of a trivariate system with three different interaction types. The first column shows example processes, the second column shows the pair correlation function and partial pair correlation function between process $X$ and process $Y$ (possibly accounting for process $Z$), the third column shows interactions between $X$ and $Z$ (accounting for $Y$) and the final column shows interactions between $Y$ and $Z$ (accounting for $Z$).}
    \label{fig:trivariateexample:pcf}
\end{figure}

\subsection{Debiasing comparison}\label{app:simulation:debiasing}
For a simple example, \cref{fig:bias} shows the average of both the simple plug in estimator and the debiased estimator as well as the true $L$ function, for the first models considered in the trivariate simulation, for the partial L function of the $X$ process with itself.
We can see that even in this simple case the biased estimate has a substantial artifact at low radii that is not present in the true or debiased version.

\begin{figure}
    \centering
    \includegraphics[width=0.5\linewidth]{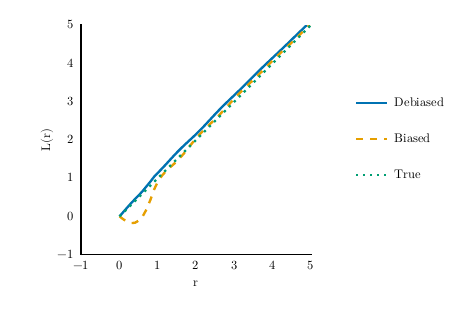}
    \caption{The average estimated $L_{XX\pred Y,Z}$ for the first trivariate model considered in the main manuscript. We show the true value, the average of the biased estimate and the average of the debiased estimate (averaged over 100 simulations).}
    \label{fig:bias}
\end{figure}

\subsection{Comparison of L function computation from spectra to direct methods}\label{app:simulation:l_function}

Another important thing to consider is the accuracy of the $L$ function computed via the spectral method compared to direct methods.
To investigate this, we compare standard border correction estimates of the $L$ function to those computed via the spectral method for standard univariate point processes.
In particular, we consider a Poisson processes, three different Thomas processes with varying cluster sizes, and a Matern hard core II process with varying inhibition radii.
The $L$ functions of each of these models are shown in \cref{fig:simulations:L_function_mse_models}.

\begin{figure}
    \centering
    \includegraphics{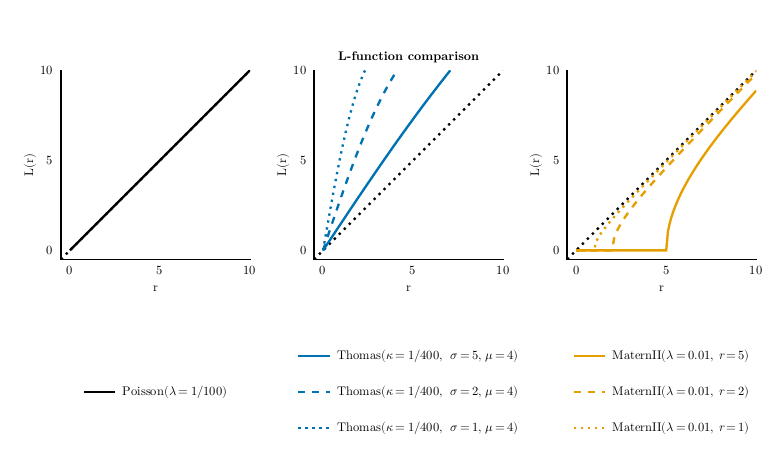}
    \caption{The $L$ functions for the models used in \cref{app:simulation:l_function} to compare the spectral and direct methods. }
    \label{fig:simulations:L_function_mse_models}
\end{figure}

We will compare the standard border corrected estimate of the $L$ function to those computed with both of the different spectral methods proposed in this work (direct and pre-rotationally averaged).
We will also vary the wavenumber grid size used in the spectral methods.
In particular, we consider grids of size $25 \times 25
, 50 \times 50, 100 \times 100,$ and a recommended grid size based on spacing scaling with the region size, as recommended in the main manuscript.
In all cases, we used a highest wavenumber based on the underlying true spectra, ensuring that we cover wavenumbers at which there is meaningful shape in the corresponding model spectra.

We then simulate 100 realisations of each model over increasing square windows of size $[0,L]^2$ for $L=100,150,200,250,300$.
We then compute average mean squared error and bias of the estimates, averaged accross all the radii considered.
We show the relative mean squared error, bias and timings in \cref{fig:simulations:L_function_mse}.
The quantities are displayed relative to the border corrected estimate, so that we show $log_{10}(\text{MSE}_{\text{method}}/\text{MSE}_{\text{border}})$ and $log_{10}(\abs{\text{Bias}_{\text{method}}}/\abs{\text{Bias}_{\text{border}}})$.

We see that in almost all cases, the spectral methods outperform the border corrected estimate in terms of mean squared error, with the exception of the Matern with the largest inhibition radius considered.
Even then, the spectral methods are competitive, which is our main goal.
Furthermore, we see that the rotationally averaged method tends to outperform the direct spectral method, particularly at smaller grid sizes.
This further justifies the use and development of this method in the main manuscript, as it is more robust in simulations than our initial approach.

\begin{landscape}
\begin{figure}[p]
    \centering
    \includegraphics[scale=0.85]{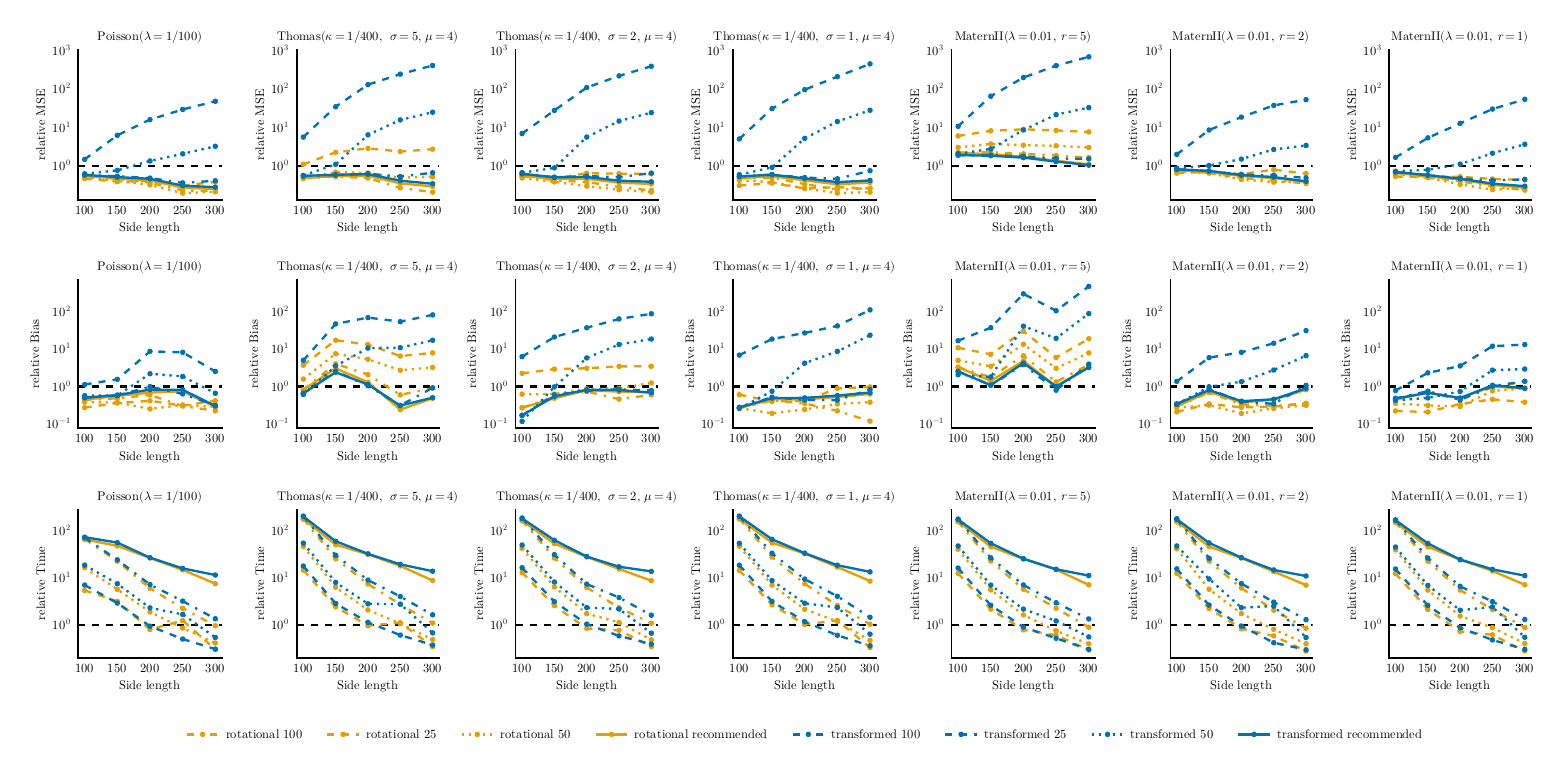}
    \caption{The relative mean squared error, bias and timings of the $L$ function estimates computed via the spectral methods compared to the border corrected estimate, for different wavenumber grid sizes. The models are the same as those shown in \cref{fig:simulations:L_function_mse_models}. }
    \label{fig:simulations:L_function_mse}
\end{figure}
\end{landscape}







\clearpage
\bibliographystyle{apalike}
\bibliography{main-bib}

\end{document}